%% file: main.tex
\documentclass[sigconf,authorversion,nonacm]{acmart}
\AtBeginDocument{%
  }

\usepackage{enumitem}
\usepackage{xspace}
\usepackage{listings}
\usepackage{anyfontsize}
\usepackage{algorithm}
\usepackage[noend]{algpseudocode}
\usepackage{varwidth}
\usepackage{multirow}

\newcommand{\name}{{\sf Helium}\xspace}
\algnewcommand{\LineComment}[1]{\State \(\triangleright\) #1}

\DeclareMathOperator{\mypath}{path}
\DeclareMathOperator{\LCApath}{LCApath}
\DeclareMathOperator{\lenOut}{len_{out}}

\definecolor{codegreen}{rgb}{0,0.6,0}
\definecolor{codegray}{rgb}{0.5,0.5,0.5}
\definecolor{codepurple}{rgb}{0.58,0,0.82}
\lstdefinestyle{codestyle}{
    commentstyle=\color{codegreen},
    keywordstyle=\color{magenta},
    numberstyle=\tiny\color{codegray},
    stringstyle=\color{codepurple},
    basicstyle=\ttfamily\footnotesize,
    breakatwhitespace=false,
    breaklines=true,
    captionpos=b,
    keepspaces=true,
    numbers=left,
    numbersep=5pt,
    showspaces=false,
    showstringspaces=false,
    showtabs=false,
    tabsize=2
}

\begin{document}

\title{Efficient LLM Serving for Agentic Workflows: A~Data~Systems~Perspective (Extended)}



\author{Noppanat Wadlom, Junyi Shen, Yao Lu}
\affiliation{%
  \institution{National University of Singapore}
  \country{Singapore}
}
\email{{noppanat, j1shen, luyao}@comp.nus.edu.sg}

\renewcommand{\shortauthors}{Wadlom et al.}

\input{sections/0_abstract}



\maketitle

\input{sections/1_introduction}

\input{sections/2_background}
\input{sections/3_overview}
\input{sections/4_optimizer}
\input{sections/5_processor}

\input{sections/6_implementation}
\input{sections/7_evaluation}
\input{sections/8_related_work}
\input{sections/9_limitations}
\input{sections/10_conclusion}



\bibliographystyle{ACM-Reference-Format}
\bibliography{references}

\input{sections/appendix}

\end{document}

%% file: sections/0_abstract.tex
\begin{abstract}
Agentic workflows are composed of sequences of interdependent Large Language Model (LLM) calls, and they have become a dominant workload in modern AI systems. These workflows exhibit extensive redundancy from overlapping prompts and intermediate results due to speculative and parallel exploration. Existing LLM serving systems, such as vLLM, focus on optimizing individual inference calls and overlook cross-call dependencies, leading to significant inefficiencies. This paper rethinks LLM and agent serving from a data systems perspective and introduces \name, a workflow-aware serving framework that models agentic workloads as query plans and treats LLM invocations as first-class operators. \name integrates proactive caching and cache-aware scheduling to maximize reuse across prompts, KV states, and workflows. Through these techniques, \name bridges classic query optimization principles with LLM serving, achieving up to {1.56×} speedup over state-of-the-art agent serving systems on various workloads. Our results demonstrate that end-to-end optimization across workflows is essential for scalable and efficient LLM-based agents.
\end{abstract}

%% file: sections/1_introduction.tex
\section{Introduction}

AI agents are autonomous LLM-based programs that act on a user’s behalf~\cite{lei2024autonomous-agents, xi2025rise-agent, luo2025agent-survey}. They operate through agentic workflows~\cite{singh2024enhancing-workflows, zhang2025aflow, yu2025workflow-survey, sapkota2025agents-vs-agentic} that are goal-driven sequences of steps involving multiple LLM invocations (often using different prompts or tools), orchestrated to solve complex tasks. They have become a dominant workload in modern AI systems~\cite{asgar2025agentic-hetero, xi2025deepsearch-survey, liu2025supporting}.

\begin{figure}[t]
    \centering
    \includegraphics[width=\linewidth]{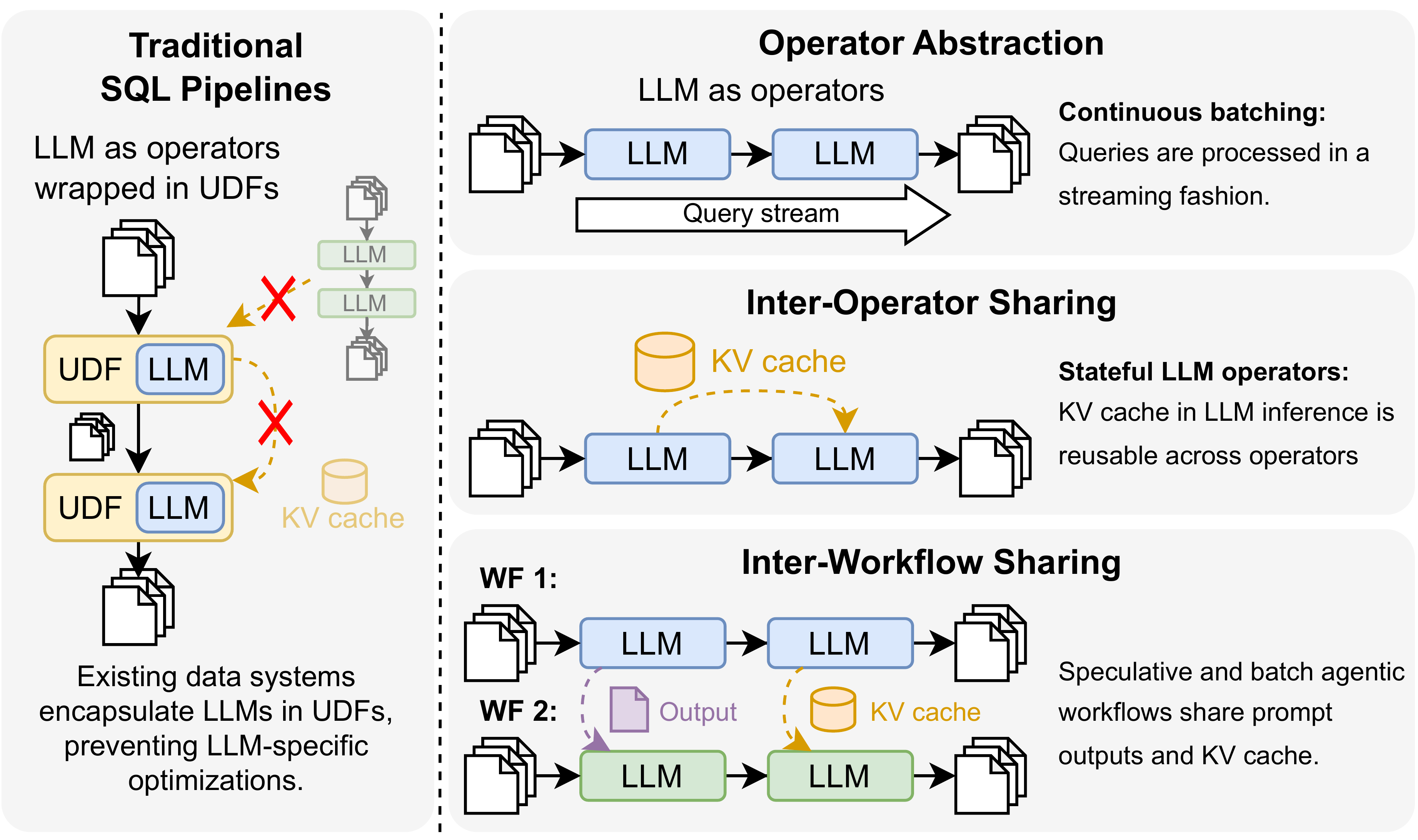}\vspace{-0.1in}
    \caption{Three disparities between traditional SQL pipelines and agentic workflows with LLM as operators.}
    \label{fig:intro}
\end{figure}

An emerging scenario of such workflows is that agents may explore many strategies or subtasks (e.g., trying alternative analyses) in parallel or in sequence to achieve their goal. This agentic speculation \cite{liu2025supporting} is a high-throughput exploration to blend LLM serving with batch-style query processing. In other words, an agentic workflow resembles a batch analytics pipeline where LLM calls function as operators. While powerful, these workflows can issue a large volume of LLM queries, often with overlapping or repeated sub-queries, leading to significant redundancy and inefficiency.

Recent advances in LLM serving systems have largely focused on optimizing individual LLM inference tasks. State-of-the-art serving engines such as vLLM~\cite{vllm} employ techniques like continuous batching~\cite{orca}, which dynamically group incoming requests to better utilize GPUs. These innovations, along with optimized GPU kernels and memory paging strategies~\cite{flashattention, flashinfer, pod-attention, fasttree}, have dramatically increased throughput and reduced latency for standalone LLM queries.
However, existing LLM inference engines operate at the granularity of individual LLM calls, lacking visibility into the broader workflow structure. These solutions are not designed for batch agentic workflows that chain up multiple LLM calls (often with different prompts).  As a result, in complex multi-LLM pipelines, e.g., an agent that spawns several helper agents, each querying an LLM, these per-call optimizations fail to capture cross-call commonalities (e.g., shared sub-prompts or intermediate results). Recent work, such as AgentScope~\cite{agentscope, agentscope-simulation}, supports basic multi-agent pipelines, but they again treat each LLM call or agent as an individual unit, optimizing locally rather than globally. This gap leaves performance optimizations unexplored for batch agentic workflows.

How can we optimize agentic LLM workflows end-to-end? Many of the challenges mirror classic problems in query processing and optimization. An agent workflow can be represented as a directed acyclic graph (DAG) of operations: nodes perform data retrieval or invoke an LLM (i.e., \emph{LLM-as-operators}), and edges represent data or prompt flow between operators. This is analogous to a complex relational query plan or a workflow in a data system. Prior work suggests that decades of database query optimization principles can be applied by modeling an agentic workflow as a query-plan DAG and applying cost-based optimization. The approach of logical and physical optimization, i.e., rewriting plans, choosing operator implementations, and scheduling execution to minimize cost, contributes to the “agent-first” query processing techniques to handle the scale and complexity of LLM-driven workloads~\cite{liu2024agentic}. In essence, this is a familiar data systems problem framed in the context of LLM serving to eliminate redundant work by globally optimizing the workflow, much like a SQL optimizer would do for a complex relational query.  
However, at the same time, important new factors differentiate LLM workflows from traditional SQL pipelines. We identify and summarize the disparities of LLM-as-operators as follows and as illustrated in Figure~\ref{fig:intro}: 

\begin{itemize}[leftmargin=*]
    \item \textbf{Operator abstraction}: The operators in an agentic workflow wrap expensive LLM inference processes. The {continuous batching} technique in LLM serving systems becomes dominant to enable processing of multiple input queries in a streaming fashion without being blocked by other queries in the same batch. Within a single LLM call, the model maintains internal state (e.g., a KV cache of the prompt) and generates output token by token in an iterative fashion. In essence, the operators' batching mechanisms are no longer abstracted by classic relational functions like select or filters.
    \item \textbf{Inter-operator sharing}: As LLMs are inherently stateful, standard data-processing optimizations must be rethought. Across multiple calls, an agent may carry forward conversational context or reuse an earlier answer in a later prompt, along with the KV cache and model's internal state, creating stateful dependencies between operations.   
    \item \textbf{Inter-workflow sharing}:  For speculative and batch agentic workflows, prefix sharing among the queries from an input batch becomes an important technique to reuse partial or even entire agentic sub-chains across queries. This is also the case for workflows across different input batches to query common data sources (e.g., agents that query daily weather or news). 
\end{itemize}

Unfortunately, current data systems are not tailored for these shifts and foresee significant challenges in efficiently processing agentic workflows. LLM operators often are wrapped in user-defined functions (UDFs) by current frameworks (e.g., Spark, Dask \cite{spark, rocklin2015dask}), while an LLM call is treated as a black-box UDF, which prevents introspection or special-case handling for performance. This lack of visibility precludes many optimizations: the system cannot automatically reuse common partial work (like shared prompt prefixes), pipeline LLM calls, or reorder operations based on LLM-specific cost factors. Importantly, continuous batching does not naturally happen beyond each individual UDF, causing blockage and significant performance degradation. Nevertheless, cost-based Volcano-style query optimization~\cite{graefe1994volcano} principles still apply, as operator costs and even semantics are often readily available, but LLM-as-operators introduces a new form of batching abstraction, statefulness, and cost that traditional optimizers were never designed to handle.

To mitigate the paradigm shifts in agentic workflows and close the gap between the continuous batching abstract in LLM serving systems and the classic batching abstract in data systems, we consider a key missing jigsaw to be a new \emph{workflow-aware LLM serving} that applies continuous batching while efficiently manages KV and prompt caches from a holistic, workflow perspective, instead of within the scope of individual UDFs. Specifically, it must simultaneously support: (i) inter-operator sharing, by enabling KV state communication across LLMs; (ii) inter-query and inter-batch sharing, by organizing caches according to prompt prefixes across multiple workflows, as well as (iii) an optimizer to maximize the chance of these sharing. 
While existing LLM serving systems (e.g., vLLM) employ prefix caching to reuse KV states from previously encountered prefixes, these approaches are inherently passive, optimized for online serving environments that cannot anticipate future workloads. In contrast, batch agentic workloads expose structures that can be leveraged for a more proactive caching strategy that exploits workload patterns to minimize redundancy. We bridge the gap between AI/MLSys and databases by adapting established query optimization principles (rather than isolated heuristics) to optimize the complex state dependencies and execution patterns of agentic workflows.

In this paper, we propose \name\footnote{Our source code is available at \url{https://github.com/mlsys-io/helium_demo}.}, a system that rethinks serving agentic workflows in modern data systems. \name introduces a novel proactive KV cache paired with a query optimizer to mitigate the drawbacks in today's passive, opportunistic KV cache sharing. By analyzing the workload during compilation, \name recognizes shared prefixes across operators and workflows, and hence pre-warms the cache to reduce redundant computation. Beyond that, \name augments and leverages a cache-aware, cost-based query optimizer to enable rewriting query plans across the workload, maximizing KV state reuse in batch agentic workloads. Compared with state-of-the-art solutions, our prototype achieves up to a 1.34× speedup on a complex financial analysis workflow, with up to a 1.56× speedup on primitive workflows.

In summary, this paper makes the following contributions:
\begin{itemize}[leftmargin=*]
    \item We present \name, a workflow-aware serving layer that models agentic workflows as query plans with LLM as first-class operators, enabling holistic intra- and inter-query optimization.
    \item We introduce a novel proactive caching strategy that pre-warms KV caches for static prompt prefixes and maintains a global prompt cache to bypass redundant operators.
    \item We design a cost-based, cache-aware scheduling algorithm that leverages a templated radix tree to capture prompt structure and dependencies, maximizing prefix cache reuse across batch agentic workloads.
    \item We implement and evaluate \name, demonstrating significant performance improvements over state-of-the-art systems across diverse workload patterns while preserving exact semantics.
\end{itemize}

%% file: sections/2_background.tex
\begin{figure*}[t]
    \centering
    \includegraphics[width=\linewidth]{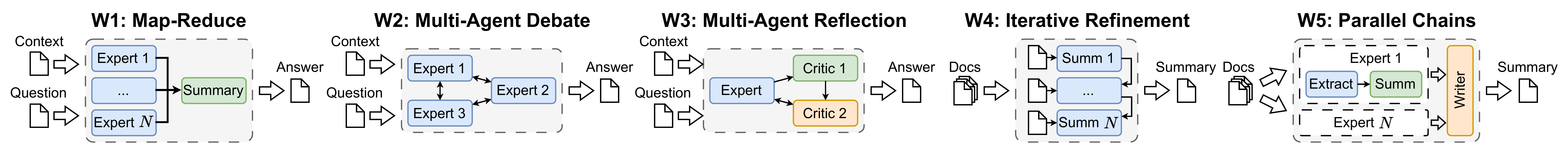}\vspace{-0.05in}
    \caption{Each representative agentic workflow demonstrates a primitive pattern in agent interactions.}
    \label{fig:workflows}
\end{figure*}

\section{Background and Motivation}
\label{sec:background}
The rise of Large Language Models (LLMs) has shifted application development towards complex, multi-step \textit{agentic workflows}~\cite{singh2024enhancing-workflows, asgar2025agentic-hetero, xi2025deepsearch-survey}. Figure~\ref{fig:workflows} demonstrates several representative workflow patterns. These workflows, which resemble traditional data processing DAGs using LLM calls as operators, often involve speculative execution, generating massive redundancy across prompts and results~\cite{wang2023selfconsistency, yao2023tree-of-thoughts, liu2025supporting}. This transforms LLM serving from a simple inference into a problem that requires optimizing computational graphs where each operator is an expensive, stateful LLM call. This section first deconstructs prior work by examining several relevant pillars.
 
\vspace{0.05in}\noindent\textbf{Optimizing LLM Serving for Latency}.
The first pillar of work focuses on optimizing the LLM operator itself in serving systems, which maximizes hardware utilization for streams of independent requests, treating each LLM call as a discrete unit of work. State-of-the-art engines like vLLM~\cite{vllm, vllm-github} use techniques like \textit{PagedAttention} to manage the key-value (KV) cache efficiently and \textit{continuous batching} to dynamically group requests, significantly improving GPU utilization and throughput for LLM serving. To reduce redundant computation, these engines also implement prefix caching, reusing pre-computed KV cache for requests that share common prompt prefixes, often managed with an LRU eviction policy~\cite{vllm, sglang, tensorrt-llm}. Aiming to improve the latency of LLM serving systems, these optimizations are fundamentally \textit{local}, \textit{workflow-agnostic}, and \textit{reactive}, as the serving engine has no visibility into the workload or the broader DAG. This "operator-level myopia" prevents it from addressing cross-call inefficiencies. For instance, it cannot guarantee KV cache reuse for related queries within a workflow if they are separated by unrelated requests, as its optimization is limited to the immediate queries. 

\vspace{0.05in}\noindent\textbf{Orchestrating LLMs as Black-Boxes}.
The next pillar comprises frameworks that orchestrate the logical composition of agentic workflows. These tools provide high-level abstractions for building complex DAGs but treat the LLM operator as a black-box unit. Frameworks like LangGraph \cite{langgraph} simplify building multi-LLM applications by properly managing data and control flows. Traditional data systems like Spark \cite{spark} integrate LLMs as User-Defined Functions (UDFs), which can be \textit{inference-agnostic}. By treating LLM invocations as black-box UDFs, the orchestration layer is blind to the internal mechanics of the LLM operator; critical performance factors, such as the stateful KV cache and the bimodal prefill/decode cost structure~\cite{distserve}, are hidden from the optimizer, preventing the system from making intelligent, cost-based decisions.
 
\vspace{0.05in}\noindent\textbf{Challenges and Our Ideas}.
Prior LLM serving and data systems were originally designed under different contracts (classic vs continuous batching). We observe the following critical challenges:

\begin{itemize}[leftmargin=*]
    \item \emph{KV cache in single operators:} In multi-agent debates \cite{multiagent-debate, liang2024mad}, each turn builds on shared conversational history, yet current systems redundantly reprocess the entire context. An ideal serving layer should instead extend the KV cache, converting costly recomputation into lightweight incremental updates. 
    \item \emph{Passive prefix caching}. Existing prefix caching is passive and opportunistic in reusing KV states, as the incoming queries are unpredictable during online serving. Prior LLM serving systems did not leverage workload patterns in batch agentic workflows.  
    \item \emph{Cost modeling and optimization:} Integrating LLMs in UDFs causes a significant disconnect: the query optimizer is blind to physical execution, while the execution engine (LLM serving systems) is blind to the logical plan.  
\end{itemize}

Rethinking the abstraction of LLM serving in data systems, our key ideas are two-fold. First, instead of employing passive, opportunistic KV cache sharing in prior LLM serving solutions, we propose a holistic, proactive caching mechanism in batch agentic workflows to reuse KV cache across operators and workflows. Next, \name augments and leverages a cost-based query optimizer to pair with the proactive cache, thus enabling query plan rewrite and maximizing cache reuse across operators and workflows. These techniques combined contribute to workflow-aware LLM serving tailored to modern data systems for agentic workflows.

\vspace{0.05in}\noindent\textbf{Scope}. We employ the following scope and simplifying assumptions in our solution: 

\begin{itemize}[leftmargin=0.15in]

\item\emph{Semantic preserving}: Optimized executions produce the same results as naive ones. We avoid approximations (e.g., proxy models) that trade accuracy for speed, focusing on unstructured data analytics rather than SQL-generation tools for structured data.

\item\emph{On-premise deployment}: We study on-prem execution with multiple GPUs, enabling fine-grained scheduling, memory, and resource control, instead of using off-the-shelf cloud APIs. 

\item\emph{Simplifying assumptions}: We limit our scope of agents to calling LLMs and performing local data operations only, without remote API calls. Also, we constrain the agentic workflows used in this work to use the same base LLM. They can be configured with different prompts to simulate different agentic roles. 
\end{itemize}

%% file: sections/3_overview.tex
\section{System Overview}

\name adopts a classic multi-stage query processing architecture~\cite{goetz1993query-evaluation, hellerstein2007db-architecture}, dividing execution into parsing, optimization, and processing phases. Each \emph{agentic workflow} is represented using a procedural language that specifies the individual LLM or agent calls and their dependencies. The system treats each workflow definition as a template that will be evaluated over a batch of input instances.  

\vspace{0.05in}\noindent\textbf{Parsing Agentic Workflow DSL}.
We use a domain-specific language (DSL) to represent batch agentic workflows that share the structure but with different input prompts. Our parser translates them into a DAG. Each operator in the DAG corresponds to a prompt-related action, such as invoking a model, retrieving data, or running a custom transformation. \name's design is inspired by the dataflow representation in TensorFlow~\cite{tensorflow} that first constructs a symbolic graph of operators, using placeholders to mark where actual prompt tokens or data will be fed in at execution time. During execution, each query and the prompts \emph{flow} through the DAG. Unlike TensorFlow, \name allows cross-operator continuous batching such that ready outputs can be forwarded to the next operators without blockage. More details are provided in Section~\ref{sec:implementation}.

\vspace{0.05in}\noindent\textbf{Logical Plan Optimization}.
\name’s query optimizer rewrites the logical DAG to eliminate redundancy and exploit shared computation across the batch workload. It applies rewrite rules to prune and merge operators. Redundant nodes (e.g., identity operators or unused branches) are removed, and identical subgraphs that occur across queries are consolidated. This is akin to common sub-expression elimination in databases, such that redundant subplans are detected and computed only once. 
After structural optimization, the optimizer performs cache substitution: for each operator, the optimizer checks if a matching entry exists in a global prompt cache that maps the inputs of deterministic operators to their outputs. On a cache hit, the corresponding LLM operator is replaced by a lightweight \texttt{CacheFetch} operator. This transformation converts a computational dependency into a simple data retrieval.

\vspace{0.05in}\noindent\textbf{Execution Planning and Processing}. \name then constructs a \emph{templated radix tree} (TRT) over the optimized logical plans to capture the prompt structure and identify commonalities. This TRT serves as the primary input for \name's cache-aware scheduling that uses a cost-based model to assign operators to workers and determine an execution order that balances load and maximizes KV cache reuse for shared prompt prefixes. \name's query processor then executes the plan with our proactive caching mechanism. Specifically, for static prompt prefixes identified by the TRT, proactive caching pre-computes and stores these states in the GPU memory during the first execution; in subsequent batches, workers can directly reuse these tensors to avoid prefill. Meanwhile, a global proactive prompt cache is maintained to store the full outputs of deterministic operators, allowing the system to bypass entire operator executions for repeated inputs. These efforts allow \name to accelerate batch agentic workflows with unchanged outputs.

\begin{figure}[t]
    \centering
    \includegraphics[width=\linewidth]{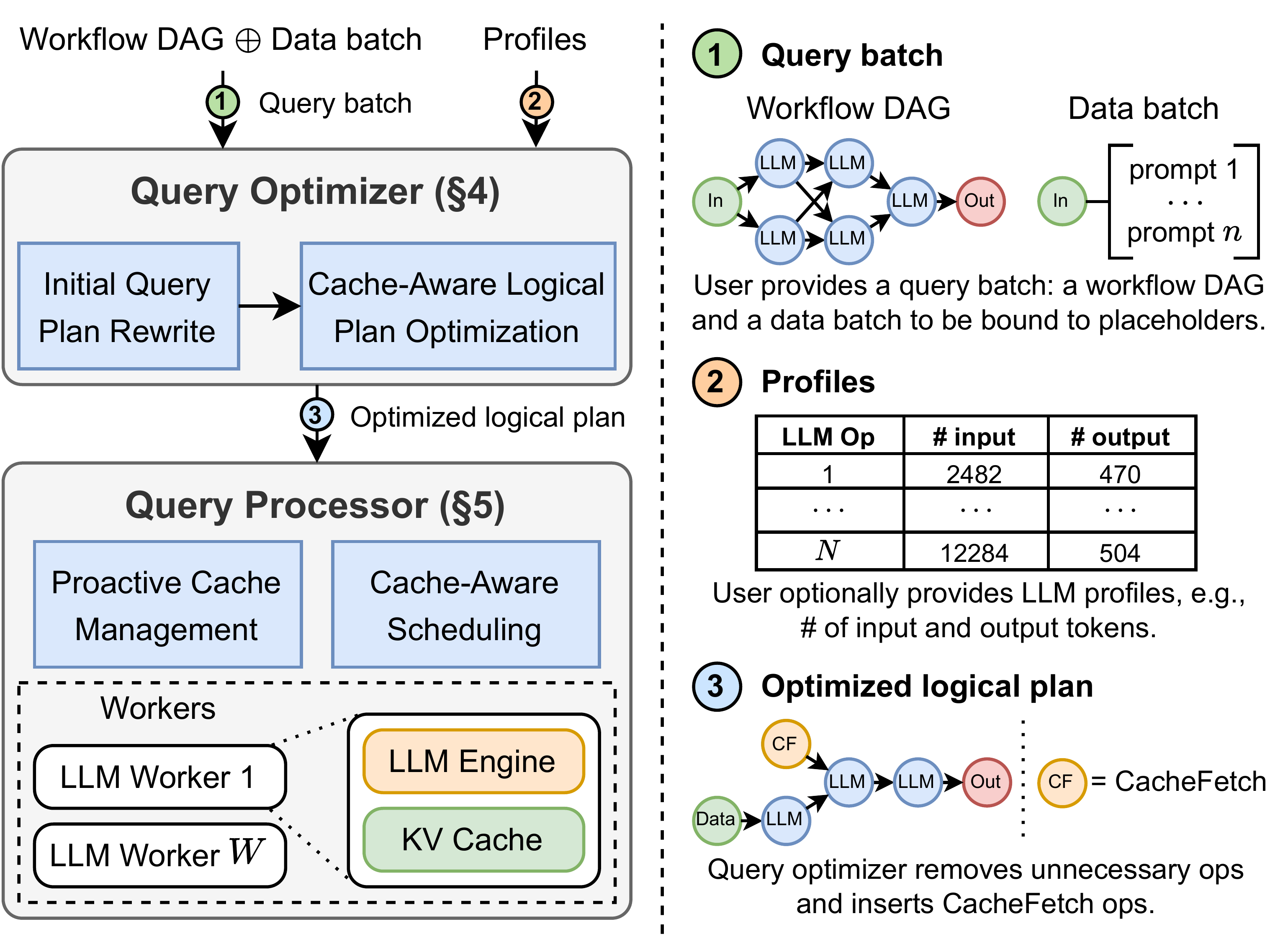}\vspace{-0.05in}
    \caption{Overview of \name's architecture.}
    \label{fig:helium-design}
\end{figure}

%% file: sections/4_optimizer.tex
\section{Query Optimizer Design}\label{sec:optimizer}
The query optimizer identifies sharing opportunities at multiple granularities, from entire sub-workflows to prompt prefixes, and constructs the logical plan for execution, bridging the gap between workflow-agnostic LLM serving engines and inference-agnostic orchestration. The optimization process comprises two stages.

\vspace{0.05in}\noindent\textbf{Initial Plan Pruning}.
User-defined agentic workflows, especially those with speculative execution, often contain structural redundancies and inefficiencies such as dead code or duplicated computations~\cite{yao2023tree-of-thoughts, besta2024graph-of-thoughts, liu2025supporting}. This stage produces a clean graph representation that simplifies the problem space for subsequent optimizations.

\begin{itemize}[leftmargin=*]
\item \emph{Operator Pruning}. Analogous to dead code elimination in compilers, this transformation removes operators that do not contribute to the final output. The optimizer performs a backward traversal from the designated output nodes, identifying and pruning any operator whose output is not consumed on a valid execution path. In multi-agent workflows, this efficiently removes speculative branches that are not rendered eventually.

\item \emph{Common Subgraph Elimination}. Next, the optimizer applies common subgraph elimination (CSE)~\cite{sellis1988optimization}, to identify and merge structurally identical subgraphs that share the same inputs. Subgraphs are hashed based on their topology and input node identifiers to detect duplicates efficiently. This ensures a computation with specific inputs is executed only once, and the prompt prefixes associated with this computation can be uniquely identified. For example, in the Map-Reduce pattern in Figure~\ref{fig:workflows}, multiple agents can be initialized with the same but repeatedly defined context. Without CSE, the preparation of this context incurs redundant computations, and the shared prompt prefixes containing this context cannot be identified.
\end{itemize}

\noindent\textbf{Logical Plan Optimization with Prompt Cache}. Next, \name leverages a global prompt cache to bypass redundant computation at the operator level, as agentic workflows often repeat entire tasks within or across queries. This cache maps the inputs of deterministic operators to their previously computed outputs. The optimizer performs a recursive, bottom-up traversal of the DAG. For each operator, it computes a signature from its type and the values of its materialized inputs (i.e., constants or values resolved from the cache); the signature is used to probe the prompt cache. Upon a cache hit, the operator is marked; during the resolution stage, the optimizer replaces it with a \texttt{CacheFetch} operator. This lightweight operator stores a pointer to the cached result. This transformation fundamentally alters the data flow; a computational dependency is rewired into a simple data retrieval dependency. Consider a \texttt{Summarizer} agent consuming output from an \texttt{Expert} agent. If the \texttt{Expert}'s input was processed previously, \name replaces the operator with a \texttt{CacheFetch}, allowing the \texttt{Summarizer} to retrieve cached output immediately and bypassing the expensive document processing entirely.

To ensure unchanged results and reproducibility, this caching mechanism is restricted to deterministic operators, such as non-LLM operators and LLM operators with greedy sampling (e.g., zero temperature), and uses LRU as the eviction policy. The optimization overhead is negligible compared to LLM inference latency. It involves a linear graph traversal consisting of lightweight CPU-bound tasks like hashing signatures and rewriting graph nodes.

\vspace{0.05in}These transformations produce an optimized \textit{logical} plan to express the \textit{intent} to fetch from a cache or execute in parallel, without binding operations to specific workers or timelines. This allows the query processor to generate physical plans based on runtime resource availability and updated cost models.

%% file: sections/5_processor.tex
\section{Query Processor Design}\label{sec:processor}\label{sec:scheduling}

\begin{figure}[t]
    \centering
    \includegraphics[width=0.8\linewidth]{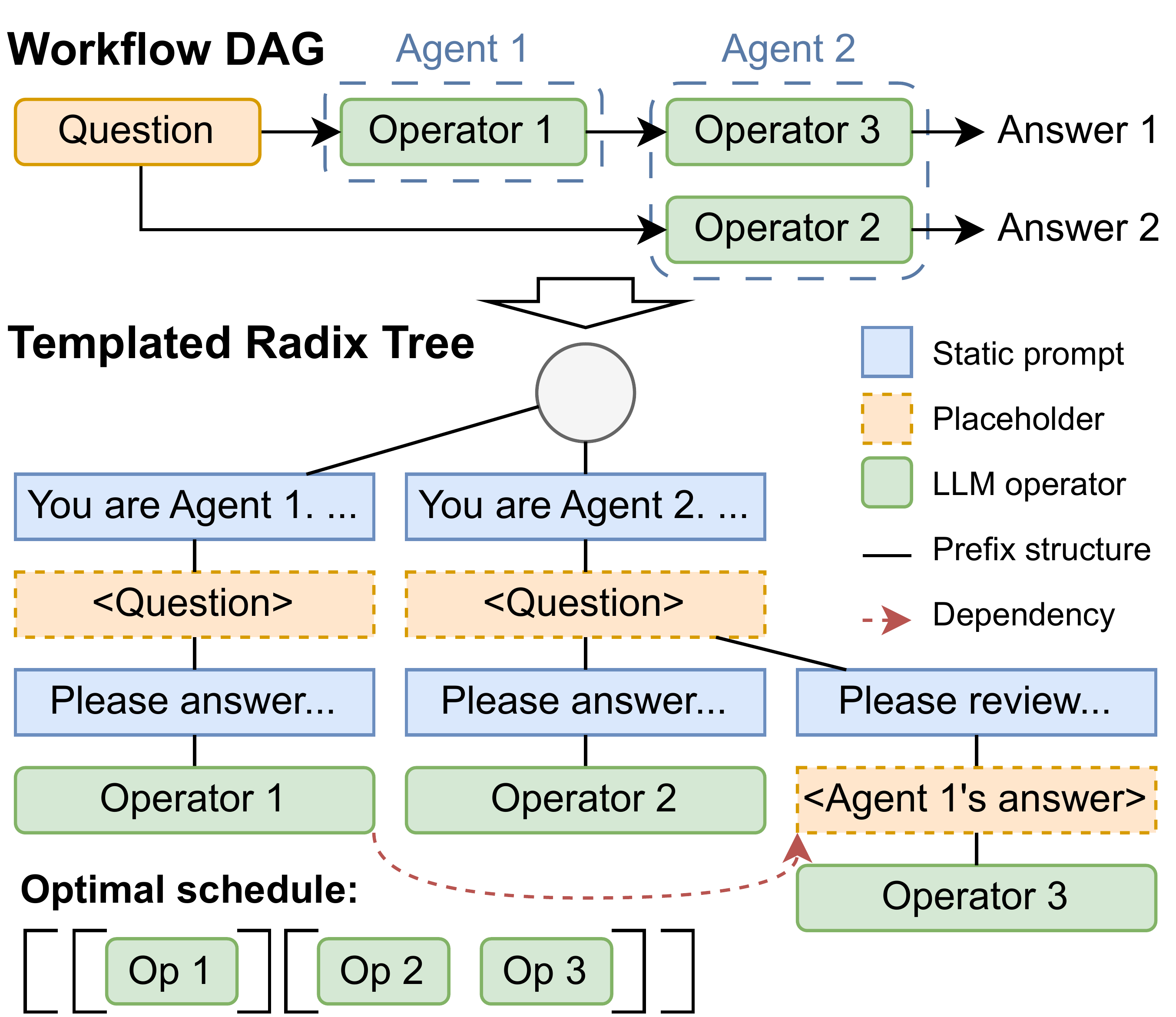}\vspace{-0.05in}
    \caption{A workflow DAG (top) and the corresponding templated radix tree with cache-aware schedule (bottom)}
    \label{fig:trt}
\end{figure}

\vspace{0.05in}\noindent\textbf{Modeling Prompt Structure with a Templated Radix Tree}.  By representing an agentic workflow as a DAG, \name can infer the structure of input prompts among operators and identify opportunities for prefix cache reuse. This is done with a \emph{templated radix tree} (TRT), a novel data structure upon the standard radix tree~\cite{morrison1968patricia}. The TRT represents the prefix structure of both static prompt components and dynamic parts derived from other operators' outputs. It also captures the dependencies among its leaf nodes.

Formally, a TRT $T = (V, E, E')$ consists of a set of nodes $V$, a set of edges $E$ representing the prefix structure, and a set of directed edges $E'$ connecting the leaves $L \subset V$. Let $r$ denote the root of $T$. Each intermediate node $v \in V \setminus (L \cup \{r\})$ is associated with a sequence of tokens and placeholders, denoting token sequences to be filled by other operators' output. Each leaf $l \in L$ is associated with an LLM operator whose input prompt structure is defined by the path from the root $r$ to its parent. The leaves $L$ and dependency edges $E'$ form a DAG $G = (L, E')$, representing the dependencies among the operators. Specifically, an edge $(l_1, l_2) \in E'$ indicates that the operator at $l_2$ depends, directly or indirectly, on the output of the operator at $l_1$. The TRT effectively captures the prefix structure of the operators' input prompts and their dependencies. We describe the algorithm for constructing the TRT in \autoref{sec:trt-construction}.

In modern LLM inference engines, the cache of common prefix tokens across different calls can be \emph{reused} to improve prefill efficiency~\cite{vllm, sglang}. Due to limited GPU memory, the prefix KV cache may be prematurely evicted to make space for newly scheduled calls. Prior work, SGLang~\cite{sglang}, proposed organizing the KV cache as a radix tree and scheduling LLM calls in order of their shared prefix length. While this strategy maximizes prefix cache reuse when all LLM calls are available upfront, it is suboptimal for agentic workflows. Consider a workflow with two agents (the top panel of \autoref{fig:trt}): Agent~1 generates an answer (Operator~\textcircled{1}); Agent~2 answers the same question while providing feedback on Agent~1's answer (Operator~\textcircled{2} and \textcircled{3}). Based on its templated radix tree, an optimal schedule would be \textcircled{1} $\rightarrow$ \textcircled{2} $\rightarrow$ \textcircled{3}, as this enables full reuse of the KV cache for Agent~2's system prompt between Operator~\textcircled{2} and \textcircled{3}. However, the online algorithm might yield a suboptimal schedule like \textcircled{2} $\rightarrow$ \textcircled{1} $\rightarrow$ \textcircled{3}. In this case, the KV cache for Agent~2's prompt might be partially evicted to process Operator~\textcircled{1}, leading to recompute and reduced efficiency for Operator~\textcircled{3}.

To address this, we must capture prefix structures across the entire workflow. Parrot~\cite{parrot} uses Semantic Variables as placeholders in prompt templates to model data dependencies. While it schedules requests to engines holding cached KV states, it encapsulates prompt structures only within individual LLM calls. This restricts the system to a dependency DAG rather than a global prefix hierarchy, forcing it into reactive scheduling that detects sharing only when requests are ready, potentially resulting in the same suboptimal schedule described above. In contrast, our TRT explicitly models both dependencies and the global prefix hierarchy, enabling proactive scheduling that minimizes execution cost.

\vspace{0.05in}\noindent\textbf{Proactive Cache Management}.\label{sec:caching}
Batch agentic workflows are highly redundant within and across batches. Queries in a batch often share the same structure, repeating prompt prefixes and intermediate results. Because agents repeatedly probe similar information, many prompts and outputs remain unchanged across runs. For example, a trading agent summarizing daily company reports sees near-identical inputs, duplicating LLM computation.
\name improves efficiency through a proactive KV caching mechanism. 

Specifically, \name’s query processor uses the TRT, constructed during scheduling, to identify static prompt prefixes that are invariant across batches. During the first execution of a workflow, the query processor precomputes and stores the KV cache for these static token sequences in GPU memory. In subsequent batches, the LLM engines can directly use these precomputed KV tensors, avoiding redundant prefill computations.

\vspace{0.05in}\noindent\textbf{Scheduling Problem Formulation}. To map each operator to the workers, we model our scheduling task as a multi-worker scheduling problem to find an assignment of calls to workers and an execution order that optimizes the \emph{makespan} (i.e., wall-clock latency) of the execution.  Specifically, \name leverages the \emph{token usage} for each LLM call, and the total \emph{token step} of the entire workflow. The former is the total number of tokens it occupies across all its inference steps. A token step serves as the unit of time in our model, where each LLM call consumes a number of token steps proportional to its token usage. This formulation allows us to model the effect of prefix cache sharing by reducing a call's token usage by the number of shared prefix tokens. To account for batching parallel calls for a single inference step, we impose \emph{precedence delay constraints} that define the minimum number of token steps that must pass between a call and another call that depends on its output. This incentivizes the scheduler to intersperse dependent calls with independent ones, thereby increasing the potential batch size.
We consider prefix sharing only in consecutive LLM calls on the same worker. Our cost model uses a call-level TRT, where each leaf corresponds to an individual LLM call rather than an operator.

Let $T = (V, E, E')$ be a TRT with root $r$ and leaf set $L$, where each $l \in L$ corresponds to an LLM call. Let $G = (L, E')$ be the dependency DAG among the calls. Each node $v \in V$ has an associated weight $\omega(v)$, where $\omega(v) = 0$ if $v \in L \cup \{r\}$, and otherwise $\omega(v)$ is the number of tokens in the prompt segment associated with $v$. Let $W$ be the set of LLM workers, and let $M_i$ be the KV cache capacity (in tokens) of worker $i \in W$. We partition the set of leaves $L$ into $|W|$ disjoint subsets $L_1, \ldots, L_{|W|}$ corresponding to the workers. We denote a permutation (schedule) of the calls assigned to worker $i$ as $\sigma_i = (l^i_1, \ldots, l^i_{|L_i|})$, where $l^i_k \in L_i$. The overall schedule is $\sigma = (\sigma_1, \ldots, \sigma_{|W|})$. The cost of each LLM call is modeled by its token usage and precedence delays to capture batching effects.

\vspace{0.05in}\emph{Token usage.} Let $u(i, j)$ be the sequence-dependent token usage of the $j$-th call, $l^i_j$, in worker $i$'s schedule. We decompose $u(i, j)$ into prefill usage $u_p(i, j)$ and decode usage $u_d(i, j)$. The prefill usage is the number of new tokens to be processed, not shared with the previous call $l^i_{j-1}$:
\begin{align*}
    u_p(i, j) = \begin{cases}
        \sum_{v \in \mypath(r, l^i_j)}{\omega(v)}, & \text{if } j = 1 \\
        \sum_{v \in \LCApath(l^i_{j-1}, l^i_j)}{\omega(v)}, & \text{otherwise}
    \end{cases}
\end{align*}
where $\mypath(v_1, v_2)$ is the set of nodes on the path from $v_1$ (exclusive) to $v_2$ (inclusive), and $\LCApath(v_1, v_2)$ is the set of nodes on the path from the lowest common ancestor of $v_1$ and $v_2$ (exclusive) to $v_2$ (inclusive). The decode token usage models the cumulative token count over all generation steps:
\begin{align*}
    u_d(i, j) = \tfrac{1}{2}\lenOut(l^i_j) (\lenOut(l^i_j) + 1)
\end{align*}
where $\lenOut(l)$ is the estimated number of output tokens for the call associated with leaf $l$. The total token usage is then:
\begin{align*}
    u(i, j) = \alpha_i (\lenOut(l^i_j) \times u_p(i, j) + u_d(i, j))
\end{align*}
where $\alpha_i$ is a normalization constant accounting for worker performance differences. If all workers use identical hardware and models, we can set $\alpha_i = 1/M_i$.

\vspace{0.05in}\emph{Precedence delay}.
To model the batching effect, we define a precedence delay $d(i, j)$. Any call that depends on the output of call $l^i_j$ must be scheduled at least $d(i, j)$ token steps after $l^i_j$ completes. This delay is given by:
\begin{align*}
    d(i, j) = \alpha_i M_i \times \lenOut(l^i_j)
\end{align*}

The total token step $T(\sigma)$ for a schedule $\sigma$ is the token step at which the last LLM call finishes, analogous to the makespan in classical scheduling problems.

Putting these together, the scheduling problem is therefore to find a schedule $\sigma$ that minimizes the total token step. Let $b(i, j)$ and $c(i, j)$ be the start and completion token steps, respectively, of call $l^i_j$. The problem is formulated as:
\begin{align*}
    \text{Minimize } T(\sigma) = \max_{\substack{i=1, \ldots, |W| \\ j=1, \ldots, |L_i|}} {c(i, j)} \text{ such that}
\end{align*}
\begin{align*}
    b(i, j) &\geq c(i', j') + d(i', j'), && \forall (l^{i'}_{j'}, l^i_j) \in E' \\
    c(i, j) &= b(i, j) + u(i, j), && i \in \{1, \ldots, |W|\}, j \in \{1, \ldots, |L_i|\}
\end{align*}
Here, the first constraint enforces dependency requirements, ensuring a call starts only after its predecessors are completed and the precedence delay has passed. The second constraint defines the completion step of each call.
The above optimization problem is NP-hard. This can be shown by a reduction from the parallel machine scheduling problem for makespan minimization, which is a well-known NP-hard problem~\cite{machine-scheduling}.

\vspace{0.05in}\noindent\textbf{Solver by Cache-Aware Scheduling}.\label{sec:scheduling-algorithm}
We propose a cost-based, cache-aware greedy algorithm to guide operator scheduling. Our approach is inspired by SGLang's DFS scheduling. However, a direct application at the LLM call level is impractical due to: (1) runtime dynamics, such as unknown output lengths and complex batching behaviors in LLM engines, make the cost model inaccurate and can lead to costly pipeline stalls; and (2) the scheduling complexity would scale with the batch size, which can be unbounded. Our algorithm instead works on an operator-level TRT derived from the logical plan; the complexity thus depends on the workflow structure rather than the batch size. It takes the optimized DAG and offline-profiled operator statistics (e.g., average output token counts) as input and produces a \emph{soft schedule}. This schedule consists of nested sequences of LLM operators for each worker, allowing reordering at runtime to adapt to system dynamics.

\begin{algorithm}[t]
    \caption{\name's Cache-Aware Scheduling}
    \label{alg:scheduling-concise}
    \begin{algorithmic}[1]
        \Function{Schedule}{workflow\_dag}
            \State partitioned\_dag $\gets$ \Call{PartitionWorkflow}{workflow\_dag}\label{alg:line:partition}
            \State tree $\gets$ \Call{BuildSchedulingTree}{partitioned\_dag}
            \While{\textbf{not} \Call{Recurse}{tree, tree.root, false}}
                \State \Call{Recurse}{tree, tree.root, true} \Comment{Force progress if stuck}
            \EndWhile
            \State \textbf{return} tree.\Call{GetFinalSchedule}{\null}
        \EndFunction
        \Statex
        \Function{Recurse}{tree, node, force}
            \If{node.is\_leaf}
                \If{tree.\Call{CanSchedule}{node, force}}
                    \State tree.\Call{ScheduleNode}{node}
                \EndIf
            \Else
                \For{\textbf{each} child \textbf{in} \Call{SelectChildren}{node, force}}
                    \If{\Call{Recurse}{tree, child, force}}
                        \State node.\Call{RemoveChild}{child}
                    \EndIf
                \State tree.\Call{UpdateState}{node.children}
                \State force $\gets$ false \Comment{Only force first child}
                \EndFor
            \EndIf
            \State \textbf{return} node.\Call{IsEmpty}{\null}
        \EndFunction
    \end{algorithmic}
\end{algorithm}

Algorithm~\ref{alg:scheduling-concise} partitions workflow operators across LLM workers, replicating operators per assignment to balance load and shrink the search space. From the partitioned DAG, we build the TRT and a mirrored scheduling tree to track dependencies and states. A DFS-style recursion chooses the next child via a critical-path heuristic \cite{critical-path}, prioritizing the subtree with the greatest aggregated dependency depth and earliest schedulable token step. At leaves, operators are scheduled when data and precedence delays allow; if blocked, we force the operator with the earliest start to ensure progress. The query processor then dispatches schedules; workers execute best-effort, buffering and issuing ready LLM calls to saturate GPUs while preserving cache-friendly ordering.

It is notable that the nested sequence structure is crucial for maximizing cache reuse across different levels of prompt sharing. For instance, in \autoref{fig:trt}, while Operators \textcircled{2} and \textcircled{3} share a static system prompt, the subsequent question prompt is unique to each query in a batch. A simple operator-by-operator schedule would thrash the cache by alternating between different queries' question prompts. Our algorithm mitigates this by grouping operators that share static prefixes into inner sequences. This structure allows workers to process LLM calls query-by-query \emph{within} an inner sequence (to reuse per-query prefix) and operator-by-operator \emph{across} inner sequences (to reuse static prefix). An intermediate buffer is used to collect operators, which are released as a complete inner sequence to the final schedule once all operators sharing the same static prompt have been emitted.
Our algorithm has a time complexity of $O(|V_{int}|\cdot c_{max}^3 + |E'| \cdot d_{max})$, where $|V_{int}|$ is the number of internal TRT nodes, $c_{max}$ is the maximum branching factor, $|E'|$ is the number of dependency edges, and $d_{max}$ is the maximum TRT depth. The proof is detailed in \autoref{sec:complexity-analysis}.

%% file: sections/6_implementation.tex
\section{Implementation}\label{sec:implementation}

\vspace{0.05in}\noindent\textbf{Agentic Workflows DSL}.
We implement the Python-based DSL that, under a lazy dataflow model, constructs a \emph{symbolic} DAG of primitive operators (listed in \autoref{sec:appendix-dsl}); operator calls record nodes and edges instead of executing. \autoref{lst:example} shows a minimal example. Each call to an \texttt{ops} primitive creates a node with its operator kind and arguments, and passing symbolic handles as arguments links nodes via dependency edges. Inputs are declared explicitly with \texttt{ops.placeholder()}, which creates a named input node that can be referenced by downstream operators (e.g., as a message argument in \autoref{lst:example}). Users author workflows by composing operators and supplying the terminal nodes to \texttt{graphs.from\_ops()}; this constructor traverses dependencies from the terminals to materialize the dependency graph. Therefore, compilation binds inputs to the graph without executing the workflow. Specifically, \texttt{compile()} binds placeholders \emph{by name} to a concrete input batch while preserving the graph’s symbolic structure, yielding a compiled graph that is ready for optimization and scheduling. Finally, \texttt{invoke()} submits the compiled graph to the runtime, which applies graph rewriting, generates a cache-aware execution plan, and dispatches work to the assigned workers (Sections~\ref{sec:optimizer} and~\ref{sec:processor}).

\lstset{style=codestyle, language=Python, caption={Example usage of \name's DSL}, label={lst:example}, deletekeywords=[2]{format,compile}\\}
\begin{figure}[t]
\begin{center}    
\begin{lstlisting}
from helium import graphs, helium, ops
# Define the agentic workflow
q = ops.placeholder("q")
answer = ops.llm([ops.Msg("user", q)])
revise_prompt = ops.fmt("Revise answer...", q, answer)
fin_answer = ops.llm([ops.Msg("user", revise_prompt)])
# Build and compile the DAG
graph = graphs.build([fin_answer]).compile(q=["How many inches is 1 meter?"])
# Execute the DAG with Helium runtime
print(helium.invoke(graph))
\end{lstlisting}
\end{center}
\end{figure}

\vspace{0.05in}\noindent\textbf{LLM Engine Integration}.
\name is built on vLLM v0.16.0~\cite{vllm-github}. Each worker manages a dedicated vLLM engine instance, each running in a separate process and communicating with the worker via IPC message queues. To implement our proactive caching strategy, we augment vLLM to support pinning the KV caches of precomputed prefixes. During a precomputation phase, the \name query processor dispatches special requests to the vLLM engine. The engine prefills these requests to populate the KV cache, retaining it in GPU memory for reuse in subsequent batches. While these precomputed caches are prioritized, we mitigate memory pressure using vLLM's native block-level eviction (LRU and longest-prefix-matching). This retains frequently accessed prefixes under contention; evicted prefixes are simply recomputed in the next batch.

\vspace{0.05in}\noindent\textbf{Job Profiling}.
\name requires statistics of LLM operators, such as their average output token counts, to make effective scheduling decisions. \name provides APIs for users to profile their agentic workflows and submit these statistics along with their jobs. For our experiments, we profile all benchmark workflows offline using a different query set. In practice, the profiling overhead is minimal compared to the execution cost with a small sample of queries. This one-time cost can be amortized across workflows later on.

%% file: sections/7_evaluation.tex
\section{Evaluation}

We evaluate \name to assess its performance and efficiency in orchestrating real-world agentic workflows. Our experiments are guided by the following targets:

\begin{enumerate}[label=RQ{\arabic*}]
    \item \label{enum:rq1} How does \name perform on microbenchmark agentic workflows that exhibit primitive patterns?
    \item \label{enum:rq2} How does \name's end-to-end performance compare to state-of-the-art solutions on complex agentic workflows?
    \item \label{enum:rq3} How do \name's key components contribute to the overall system performance? We conduct an ablation study. 
    \item \label{enum:rq4} How does \name perform under different configurations and workload constraints? We conduct a sensitivity study.
\end{enumerate}

In the following subsections, we first investigate \ref{enum:rq1} using five primitive agentic workflows illustrated in \autoref{fig:workflows}. We then evaluate our solution on more sophisticated workflows that combine multiple primitive patterns to demonstrate \ref{enum:rq2}. This is followed by ablation and sensitivity studies to cover \ref{enum:rq3} and \ref{enum:rq4}.

\begin{figure*}[t]
    \centering
    \includegraphics[width=\linewidth]{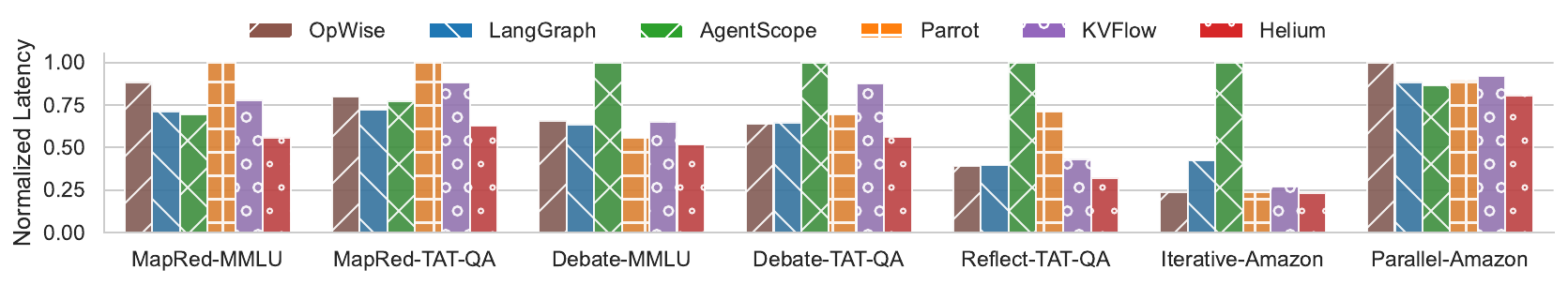}\vspace{-0.15in}
    \caption{Normalized end-to-end latency of \name and baselines, excluding \emph{vLLM}, across representative workflows and datasets with Qwen3-8B. Values are normalized within each workload so that 1.0 equals the slowest system (lower is better).}
    \label{fig:single-workflow-perf}
\end{figure*}

\subsection{Microbenchmarks}
\label{sec:microbenchmarks}
\vspace{0.05in}\noindent\textbf{Primitive Workflows}.
We consider five primitive agentic workflows illustrated in \autoref{fig:workflows}.  \autoref{tab:single-workflow} illustrates more details on the configurations. 
\begin{itemize}[leftmargin=*]
    \item Map-Reduce (\emph{MapRed}): This workflow consists of multiple expert agents that concurrently process input contexts and questions, followed by a summarizer agent that aggregates the experts' outputs to produce the final answer.
    \item Multi-Agent Debate (\emph{Debate}): Following \cite{multiagent-debate}, this workflow involves multiple agents debating over a given context and question before arriving at a final answer.
    \item Multi-Agent Reflection (\emph{Reflect}): This workflow uses an expert agent to draft an answer, followed by multiple critic agents that provide feedback to refine the initial answer.
    \item Iterative Refinement (\emph{Iterative}) \cite{iterative-sum}. Input documents are divided into smaller chunks, and a summarizer agent processes each chunk iteratively, refining the previous summary before producing the final result.
    \item Parallel Chains (\emph{Parallel}). The workflow consists of multiple expert agents with distinct roles that extract insights from separate input chunks (e.g., customer reviews), followed by a writer agent that aggregates the insights into a coherent report.
\end{itemize}
 We use datasets from MMLU \cite{mmlu}, TAT-QA \cite{tatqa}, and Amazon Reviews \cite{amazon-reviews}. For \emph{MapRed} and \emph{Debate}, we sample 200 questions from MMLU and 100 contexts from TAT-QA (each with 6 questions). For \emph{Reflect}, we use 200 contexts from TAT-QA (each with 6 questions). For \emph{Iterative} and \emph{Parallel}, we sample 200 and 100 items from Amazon Reviews, respectively, each with 60 reviews divided into 6 chunks. Since we evaluate performance on a single batch, only proactive KV caching and cache-aware scheduling are enabled.

\begin{table}[t]
    \footnotesize
    \centering
    \begin{tabular}{p{0.14\columnwidth}p{0.20\columnwidth}p{0.53\columnwidth}}
        \toprule
        \textbf{Workflow} & \textbf{Datasets} & \textbf{Configuration} \\
        \midrule
        MapRed & MMLU, TAT-QA & 14 experts for MMLU, 7 experts for TAT-QA with 7 distinct roles \\
        Debate & MMLU, TAT-QA & 3 agents with distinct roles, 2 rounds of debate \\
        Reflect & TAT-QA & 1 expert, 2 critics \\
        Iterative & Amazon & 6 review chunks, 10 reviews each\\
        Parallel & Amazon & 7 experts, each processing 6 review chunks \\
        \bottomrule
    \end{tabular}
    \caption{Configurations for primitive workflows.}
    \label{tab:single-workflow}
\end{table}

\vspace{0.05in}\noindent\textbf{Baselines}. We compare \name against baselines and state-of-the-art LLM serving systems and agentic workflow orchestration:
\begin{itemize}[leftmargin=*]
    \item \emph{vLLM}~\cite{vllm}: a high-throughput LLM inference engine. This baseline implements agentic workflows naively by executing each query's operators \emph{sequentially} on an unmodified vLLM backend. This setup is the same as that in \name to ensure fairness, serving as a reference point for performance in the absence of workflow orchestration.
    \item \emph{OpWise}: executes the workflow DAG \emph{operator by operator} across the batch. We implement this baseline to simulate the execution model of classical batch analytics systems (e.g., Spark, Dask~\cite{spark, rocklin2015dask}). Unlike \emph{vLLM}’s query-wise execution, \emph{OpWise} traverses the DAG in topological order and executes the operator for each query before moving to the next, improving data parallelism.
    \item \emph{LangGraph}~\cite{langgraph}: a graph-based orchestration framework for stateful agents. We map our workflows directly to its graph abstraction and execute on the full batch using LangChain’s Runnable interface for asynchronous batch execution~\cite{langchain}.
    \item \emph{AgentScope}~\cite{agentscope}: a multi-agent platform with an \emph{actor-based distributed} execution mechanism that parallelizes agent runs and exchanges messages among decentralized agent servers~\cite{agentscope-simulation}. We implement the workflows in AgentScope v0.1.6, passing the entire data batch as the initial inputs.
    \item \emph{Parrot}~\cite{parrot}: an LLM service system that exposes application-level knowledge via \emph{Semantic Variables} and performs \emph{dataflow analysis} across requests. For each workflow, we submit all requests for all agents and batch items up front, allowing Parrot to optimize over the whole workload.
    \item \emph{KVFlow}~\cite{kvflow}: a KV cache management framework that abstracts workflows as \emph{Agent Step Graphs} to prevent premature cache eviction and enable efficient prefetching. We employ the paper's SGLang-based implementation as the inference engine and LangGraph for workflow orchestration. Following the paper's methodology, we pre-populate the KV cache with static prompts before each experimental trial.
\end{itemize}

\vspace{0.05in}\noindent\textbf{Models and Testbed}. Our evaluation employs the Qwen3-8B and Qwen3-14B models \cite{qwen3}. All experiments are conducted on a machine equipped with an AMD EPYC 9554 64-Core Processor and two 94GB NVIDIA H100 NVL GPUs. Apart from \emph{KVFlow}, we use vLLM v0.16.0~\cite{vllm}, with automatic prefix caching and chunked prefill enabled, as the LLM inference engine throughout the experiments. Each model instance is deployed on a single GPU by spawning a separate engine instance. 

Since the baseline frameworks do not natively support integration with the inference engines, we deploy a standalone LLM inference server on the same machine and interface with the baselines through the OpenAI API. For multi-engine experiments, because the baselines, except for \emph{Parrot}, lack built-in request routing, we implement a lightweight load balancer that routes requests across worker engines based on request queue length. For reproducibility, we employ greedy sampling for all LLM calls across all experiments.

\vspace{0.05in}\noindent\textbf{Evaluation metrics}. We evaluate all frameworks using \textit{end-to-end latency} (total wall-clock time for query batch preparation and execution). For \name, this includes optimization, scheduling, and processing. For \emph{Parrot}, it includes the overhead of registering Semantic Variables and semantic functions, while for other baselines, the latency reflects only the workflow execution time. Since throughput is inversely proportional to latency for a fixed batch size, we report only latency to simplify comparison.

\begin{table}[t]
    \footnotesize
    \centering
    \begin{tabular}{lccc}
    \toprule
    \multirow{2}{*}{\textbf{System}} & \multicolumn{3}{c}{\textbf{Relative Latency (× vs. \name)}} \\
    \cmidrule(lr){2-4}
     & \textbf{Average} & \textbf{Min} & \textbf{Max} \\
    \midrule
    vLLM       & 66.27 & 38.18 & 100.92 \\
    OpWise     & 1.25 & 1.02 & 1.58 \\
    LangGraph  & 1.28 & 1.09 & 1.83 \\
    AgentScope & 2.10 & 1.07 & 4.32 \\
    Parrot     & 1.44 & 1.02 & 2.21 \\
    KVFlow     & 1.32 & 1.14 & 1.56 \\ \hline
    \name      & \textbf{1.00} & \textbf{1.00} & \textbf{1.00} \\
    \bottomrule
    \end{tabular}
    \caption{Relative latency of each system normalized to \name across representative workflows (lower is better).}
    \label{tab:single-workflow-summary}
\end{table}

\begin{figure*}[t]
    \centering
    \includegraphics[width=\linewidth]{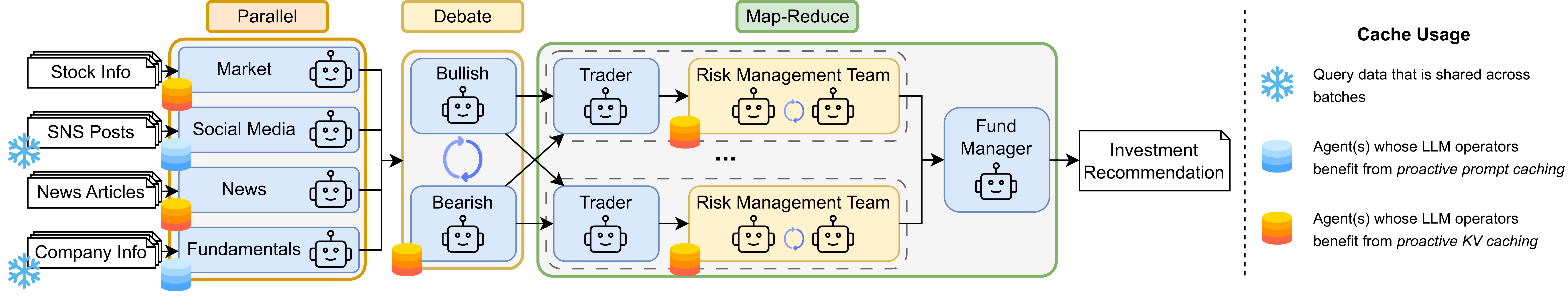}\vspace{-0.15in}
    \caption{The \emph{Trading} workflow used for end-to-end evaluation, combining the \emph{Parallel}, \emph{Debate}, and \emph{Map-Reduce} patterns. Agent annotations indicate opportunities for proactive KV or prompt caching, applied to prefixes over 200 tokens.}
    \label{fig:composite-workflow}
\end{figure*}

\vspace{0.05in}\noindent\textbf{Comparisons}. The results are shown in \autoref{fig:single-workflow-perf} and summarized in \autoref{tab:single-workflow-summary}, demonstrating \name's robust performance across all representative workflows. As expected, \name achieves a speedup of up to 100.92× over the naive \emph{vLLM} baseline, attributed to the baseline's execution model, which processes requests sequentially and precludes the throughput gains from batch computation. \name is up to 1.58× faster than \emph{OpWise}, whose operator-by-operator execution model becomes a bottleneck in highly parallel workflows (\emph{MapRed}, \emph{Parallel}) and is inefficient for workloads with cache patterns that favor a query-by-query execution order. \name outperforms \emph{LangGraph} by up to 1.83× on workflows with high operator parallelism or long dependency chains (\emph{MapRed}, \emph{Iterative}), exposing the limitations of generic graph execution. \name is up to 4.32× faster than \emph{AgentScope}, whose agent-level parallelism is ill-suited for workflows where a few agents are invoked repeatedly (\emph{Debate}, \emph{Reflect}, \emph{Iterative}). Furthermore, \name is up to 2.21× faster than \emph{Parrot}. Although \emph{Parrot} is prefix-aware, its scheduling heuristics cause severe load imbalance on workflows with specific prompt patterns (e.g., \emph{MapRed}, \emph{Debate} with TAT-QA, and \emph{Reflect}). Finally, \name outperforms \emph{KVFlow} by up to 1.56×. While \emph{KVFlow} leverages static prompt precomputation, advanced cache eviction, and hierarchical prefetching, \name's proactive KV caching and cache-aware scheduling still offer a significant performance advantage, especially on high prefix sharing scenarios like \emph{MapRed} and \emph{Debate} with TAT-QA. \name outperforms the baselines in all cases.

\subsection{End-to-End Benchmark}
\label{sec:end-to-end}

\vspace{0.05in}\noindent\textbf{Setups}. To demonstrate \name's performance on a realistic task and answer \ref{enum:rq2}, we construct a complex agentic workflow named \emph{Trading}, combining multiple primitive patterns from the microbenchmark. The workflow aims for investment recommendation and combines \emph{MapRed}, \emph{Debate}, and \emph{Parallel}, mirroring the workflow of a financial trading application proposed in prior work~\cite{tradingagents}. Evaluating this composite workflow allows us to assess \name's ability to orchestrate heterogeneous structures and exploit diverse reuse opportunities within a single workload, which isolated primitive patterns cannot capture. Concretely, the \emph{Trading} workflow has three stages: \emph{analyst}, \emph{research}, and \emph{decision}. The analyst stage follows the \emph{Parallel} pattern, with four agents (market, social media, news, and fundamentals) that analyze and summarize different documents; i.e., the market analyst processes stock price history, while the fundamentals analyst reviews company financial statements. The research stage uses the \emph{Debate} pattern, where two researcher agents debate the analysts' findings, mediated by a manager. The final decision stage resembles the \emph{MapRed} pattern but with nested complexity. It branches into eight chains, each containing a trader agent with a specific behavior (e.g., value trader, news trader \cite{harris2002trade-behavior}). The trader makes an initial decision, which is then evaluated by three risk management agents in a multi-turn debate. A fund manager agent aggregates the outputs from all chains to produce the final investment recommendation. In total, the workflow consists of 19 agents, broken down into 88 LLM operators.

\begin{figure}[t]
    \footnotesize
    \centering
    \includegraphics[width=\linewidth]{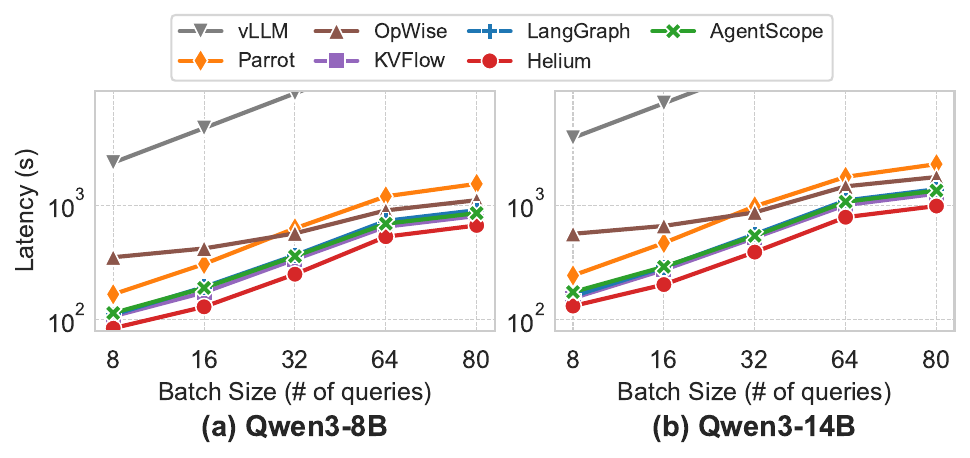}\vspace{-0.15in}
    \caption{End-to-end latency of \name and baselines on the \emph{Trading} workflow across batch sizes using (a) Qwen3-8B and (b) Qwen3-14B.}
    \label{fig:end-to-end-perf}
\end{figure}

We use the same models and experiment environment as the microbenchmark. Inspired by the task in \cite{tradingagents}, we created a financial dataset using data from a few finance data sources \cite{finnhub,yfin,reddit-finance}. We collected data for 100 stocks over two consecutive days. The first day's data is used to warm up the system, while the second day's data is used for evaluation. The dataset consists of multiple components, such as the company profiles, stock price histories, and related social media posts, each fed to the corresponding analyst agents. This setup simulates a daily trading scenario where some data, such as company fundamentals, is relatively static across batches. The workflow, with its branching and repeated sub-tasks, simulates a complex and computationally intensive workload.

\vspace{0.05in}\noindent\textbf{Comparisons}. The end-to-end latency results are summarized in \autoref{fig:end-to-end-perf}. \name achieves a significant performance improvement over the naive \emph{vLLM} baseline, reducing latency by up to 39.50×. This substantial gap highlights the critical need for a workflow-aware orchestration layer that can exploit parallelism. Compared with \emph{OpWise}, which models the workflow as a DAG but executes it operator-by-operator, \name still achieves up to a 4.25× speedup. \emph{OpWise}'s scheduling model is limited; it cannot parallelize independent LLM requests and results in poor prefix cache utilization. \name achieves speedups of up to 1.49× over \emph{LangGraph} and up to 1.46× over \emph{AgentScope}. Despite the fact that our load balancer allows these frameworks to achieve a high degree of parallelism, they are unaware of redundant requests or shared prompt prefixes. \name's proactive caching and cache-aware scheduling allow it to exploit these redundancies. Notably, \name achieves speedups of up to 2.51× compared to \emph{Parrot}. Although \emph{Parrot} also performs workflow analysis, it is not designed for batch agentic workflows and makes suboptimal scheduling decisions. For instance, its heuristic of dispatching requests to engines based on cached prefixes leads to severe load imbalance and a significant performance penalty, highlighting the importance of \name's cache-aware scheduling. Finally, \name outperforms \emph{KVFlow} by up to 1.34×. While \emph{KVFlow}'s workflow-aware eviction and hierarchical prefetching improve efficiency, it lacks the global optimizations employed by \name, resulting in the observed performance gap. More discussion on prefix cache utilization can be found in \autoref{sec:prefix-cache-util}.

\subsection{Ablation Study}
\label{sec:ablation}

To validate the individual contributions of \name's key components (\ref{enum:rq3}), we perform an ablation study on the \emph{Trading} workflow with a batch size of 16 using Qwen3-8B. We evaluate four variants of \name: one without proactive KV caching (w/o KV), one without initial plan pruning (w/o PP), one without prompt caching (w/o PC), and one without cache-aware scheduling (w/o CAS). The results, summarized in Table~\ref{tab:ablation}, show the performance delta relative to the full system. We observe that disabling plan pruning causes the largest performance drop (23.35\%). Although the \emph{Trading} workflow has no unused operators, it contains redundant operators that are otherwise removed by common subexpression elimination (CSE). Without pruning, these redundant operators persist, obscuring prefix identification and leading to a sub-optimal schedule. Disabling cache-aware scheduling causes the second-largest drop (17.66\%), emphasizing the value of intelligent task ordering for cache reuse. Removing prompt caching increases latency by 13.56\%, showing that avoiding redundant operator runs is key to efficiency. Proactive KV caching has a smaller but still meaningful effect (3.55\% delta), confirming that reusing precomputed KV states for static prefixes reduces overhead. Overall, these results demonstrate that \name's gains stem from the synergy of its caching mechanisms, query optimizer, and cache-aware scheduling.

\begin{table}[t]
    \footnotesize
    \begin{tabular}{llcc}
    \toprule
    \textbf{System} & \textbf{Configurations} & \textbf{Latency (s)} & \textbf{Perf. delta (\%)} \\ \midrule
    \name & Full & 130.14 & 0.00 \\ \cline{2-4}
    & w/o KV & 134.76 & -3.55 \\
    & w/o PP & 160.53 & -23.35 \\
    & w/o PC & 147.79 & -13.56 \\
    & w/o CAS & 153.12 & -17.66 \\ \bottomrule
    \end{tabular}
    \caption{Ablation study comparing \name's performance without proactive KV caching (KV), plan pruning (PP), prompt caching (PC), and cache-aware scheduling (CAS).}
    \label{tab:ablation}
\end{table}

\begin{table}[t]
    \footnotesize
    \begin{tabular}{cll}
    \toprule
    \textbf{Scheduling method} & \multicolumn{1}{c}{\textbf{Latency (s)}} & \multicolumn{1}{c}{\textbf{Cache hit rate (\%)}} \\ \midrule
    QueryWise & 658.54 (4.40×) & 42.3 (-25.1\%) \\
    OpWise & 419.03 (2.80×) & 40.8 (-27.8\%) \\
    Random & 193.54 (1.29×) & 37.1 (-34.4\%) \\
    LSPF & 188.93 (1.26×) & 37.9 (-32.9\%) \\ \midrule
    \name (w/ only CAS) & \textbf{149.76} & \textbf{56.5} \\ \bottomrule
    \end{tabular}
    \caption{Performance of different scheduling strategies on the \emph{Trading} workflow. Caching is disabled to isolate scheduling effectiveness.}
    \label{tab:scheduling-effectiveness}
\end{table}

\vspace{0.05in}\noindent\textbf{Scheduling Effectiveness}.
Next, we study the effectiveness of \name's cost-based cache-aware scheduling by comparing it against four well-known scheduling strategies. To isolate the impact of scheduling, we disabled proactive KV and prompt caching. The baselines include: \emph{QueryWise}, which executes queries sequentially (implemented via LangGraph without batching); \emph{OpWise}, which executes operator-by-operator across the batch; \emph{Random}, which dispatches any ready request, reflecting LangGraph's default batch execution; and \emph{LSPF} \cite{sglang}, an online prefix-aware strategy, which we implemented by modifying vLLM to sort its request queue.

The results in \autoref{tab:scheduling-effectiveness} show that \name significantly outperforms all strategies. It achieves 4.40× and 2.80× speedups over \emph{QueryWise} and \emph{OpWise}, respectively; the former fails to exploit inter-query parallelism, while the latter thrashes the prefix cache by ignoring shared prefixes across operators. While \emph{Random} maximizes parallelism, \name achieves a 1.29× speedup by optimizing for both concurrency and cache reuse. More importantly, \name is 1.26× faster than \emph{LSPF}. \emph{LSPF}'s online, workflow-agnostic approach yields only a 1.5\% hit-rate improvement over \emph{Random}, whereas \name's global TRT-based optimization improves hit rates by 32.9\% over \emph{LSPF}, proving the necessity of informed global scheduling.

\begin{table}[t]
    \footnotesize
    \begin{tabular}{ccccccc}
    \toprule
    \multirow{2}{*}{\textbf{Sched. method}} & \multicolumn{3}{c}{\textbf{Total token step ($\times 10^6$)}} & \multicolumn{3}{c}{\textbf{Optimality gap (\%)}} \\ \cmidrule(lr){2-4} \cmidrule(lr){5-7}
    & \textbf{Avg} & \textbf{Min} & \textbf{Max} & \textbf{Avg} & \textbf{Min} & \textbf{Max} \\ \midrule
    QueryWise & 21.6 $\pm$ 8.6 & 10.2 & 34.2 & 72.4 $\pm$ 39.3 & 26.7 & 149.2 \\
    OpWise & 15.3 $\pm$ 6.9 & 7.0 & 26.4 & 17.6 $\pm$ 9.9 & 3.6 & 30.7 \\
    Random & 15.2 $\pm$ 7.2 & 7.0 & 27.7 & 16.3 $\pm$ 8.1 & 6.0 & 30.3 \\
    LSPF & 14.9 $\pm$ 6.9 & 7.0 & 26.9 & 14.5 $\pm$ 8.7 & 2.8 & 30.5 \\ \midrule
    \name & \textbf{13.4 $\pm$ 6.5} & \textbf{6.0} & \textbf{24.1} & \textbf{0.9 $\pm$ 1.4} & \textbf{0.0} & \textbf{3.6} \\ \bottomrule
    \end{tabular}
    \caption{Total token steps and optimality gaps of \name's scheduling algorithm compared to the baselines.}
    \label{tab:optimality-study-results}
\end{table}

\vspace{0.05in}\noindent\textbf{Scheduling Optimality}.
We assess the optimality of \name's scheduling algorithm by benchmarking it against the theoretical optimum derived from our problem formulation (Section~\ref{sec:processor}). We reformulate the scheduling task as a Mixed Integer Linear Program (MILP) and utilize an off-the-shelf solver~\cite{forrest2024cbc} to find the schedule that minimizes the total token step cost. We evaluate 7 scaled-down dataset-workflow configurations from Section~\ref{sec:microbenchmarks} (2-4 agents processing a batch of 2-4 queries) on a single LLM worker with an 8192-token KV cache limit, setting a 6-hour solver timeout per instance. Performance is quantified using the \emph{optimality gap}, defined as the percentage cost increase of a schedule $\sigma$ relative to the reference optimal schedule $\sigma^*$: $\text{Gap (\%)} = \frac{T(\sigma) - T(\sigma^*)}{T(\sigma^*)} \times 100$.

The results, detailed in \autoref{tab:optimality-study-results}, demonstrate that \name achieves near-optimal performance with an average optimality gap of 0.9\% and a maximum of 3.6\%, significantly outperforming all baselines. \emph{QueryWise} exhibits extreme suboptimality (up to 149.2\%) due to its failure to exploit inter-query parallelism. While \emph{OpWise} and \emph{Random} leverage parallelism, their lack of prefix awareness results in gaps exceeding 30\%. Similarly, although \emph{LSPF} improves upon \emph{Random} via online prefix matching, its myopic approach still yields sub-optimal schedules with gaps surpassing 30\%, standing in sharp contrast to \name's 3.6\%. These results confirm that \name's global cache-aware scheduling is essential for approaching theoretical optimality.

\subsection{Sensitivity Analysis}

To address \ref{enum:rq4}, we conduct a sensitivity analysis to evaluate \name's performance under various workload and system configurations. For these experiments, we compare \name against \emph{LangGraph}, a strong baseline in our prior experiments.

\begin{figure}[t]
    \centering
    \includegraphics[width=\linewidth]{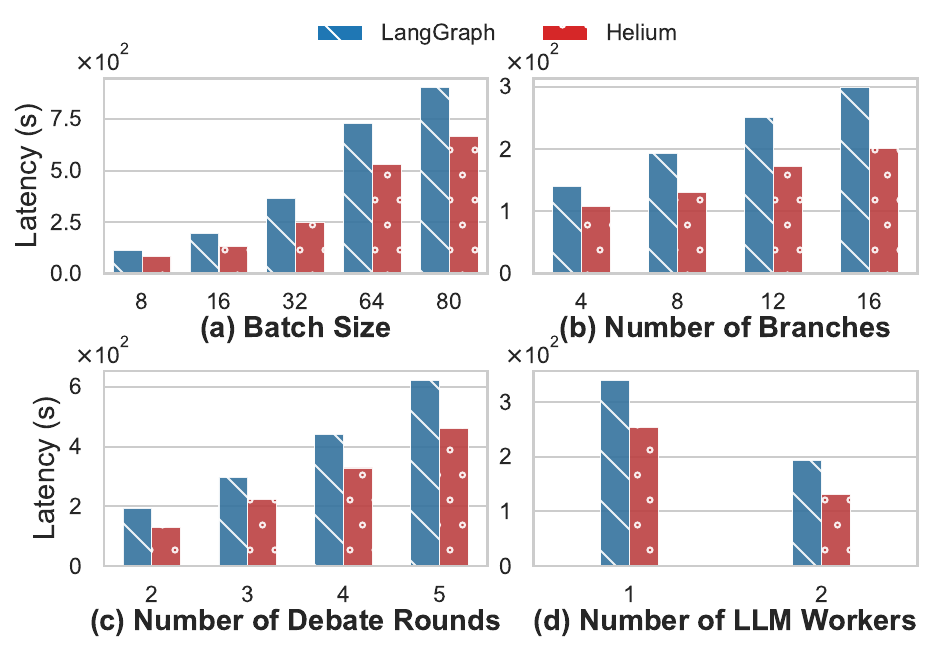}\vspace{-0.15in}
    \caption{Scalability of \name vs. \emph{LangGraph} on \emph{Trading} workflow across varying: (a) batch sizes, (b) parallel branches, (c) debate rounds, and (d) LLM workers.}
    \label{fig:scalability}
\end{figure}

\vspace{0.05in}\noindent\textbf{Workload Scalability}. We first analyze how \name's performance scales with different workload characteristics, as shown in \autoref{fig:scalability}. Using the \emph{Trading} workflow with Qwen3-8B, we vary the batch size from 8 to 80 (\autoref{fig:scalability}a). \name's performance advantage widens with larger batches, demonstrating that its caching and scheduling strategies effectively exploit the increased inter-query sharing opportunities in larger workloads, whereas \emph{LangGraph}'s black-box model cannot leverage such cross-query optimizations.

Next, we scale the number of parallel branches in the decision stage from 4 to 16 (\autoref{fig:scalability}b). \name consistently maintains its performance advantage across all configurations, efficiently handling the computational demands of highly parallel workflows. We also vary the number of debate rounds in the research stage and risk management module from 2 to 5 (\autoref{fig:scalability}c). As the workflow becomes more sequential with additional rounds, \name's performance advantage continues to grow due to its effective cache-aware scheduling that maximizes reuse across dependent operators. Furthermore, since complex workflows are typically compositions of these primitive patterns, these results demonstrate \name's ability to scale with increasing workflow complexity and adapt to a broad range of workflows involving larger operation counts. Finally, we evaluate scalability with varying numbers of LLM workers (\autoref{fig:scalability}d). While both systems benefit from additional computational resources, \name achieves superior scaling efficiency through its cache-aware scheduler, which optimizes operator placement to maximize cache reuse and minimize end-to-end latency.

\begin{figure}[t]
    \centering
    \includegraphics[width=\linewidth]{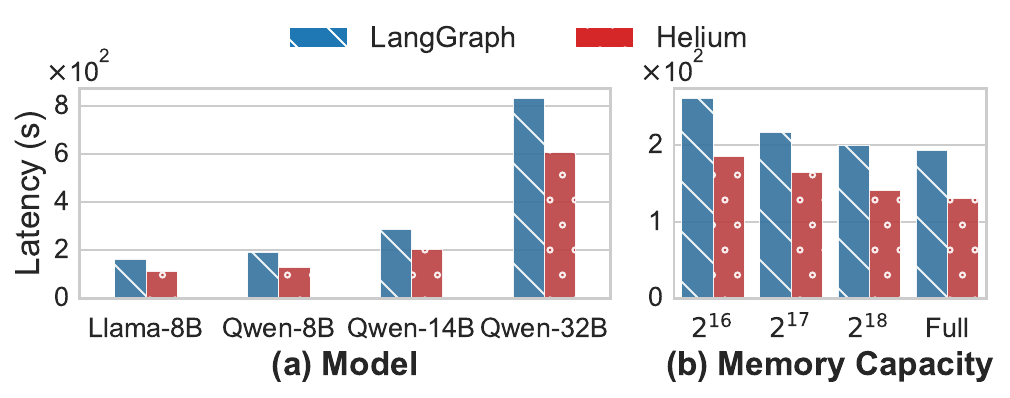}\vspace{-0.15in}
    \caption{Sensitivity to (a) LLM models and (b) KV cache capacity (tokens) on \emph{Trading} workflow. "Full" indicates KV cache occupying the entire GPU.}
    \label{fig:system-variations}
\end{figure}

\vspace{0.05in}\noindent\textbf{System and Model Variations}. We then assess \name's robustness to changes in the underlying system and LLM, with results in \autoref{fig:system-variations}. We evaluate performance on the \emph{Trading} workflow with four models: Llama-3.1-8B, Qwen3-8B, Qwen3-14B, and Qwen3-32B (\autoref{fig:system-variations}a). As expected, latency increases with model size for both systems. However, \name's performance advantage also grows. The cost of redundant computation, which \name is designed to remove, is greater with larger models. \name's advantage grows because it achieves greater savings as the underlying inference becomes more resource-intensive. In addition, larger models increase memory pressure due to a larger KV cache size per token. \name's scheduler helps with this by maximizing prefix reuse, making better use of the limited cache space.

Next, we simulate environments with constrained GPU memory by artificially limiting the KV cache size (\autoref{fig:system-variations}b). As memory pressure increases, both systems experience performance degradation from more frequent cache evictions. However, \name is more resilient; its cache-aware scheduler anticipates this pressure and organizes the execution plan to optimize cache reuse, showing better performance in resource-constrained systems.

\vspace{0.05in}\noindent\textbf{Prefix Sharing Sensitivity}. Finally, we investigate \name's robustness under varying degrees of prefix sharing. We evaluate end-to-end latency on the \emph{MapRed} and \emph{Debate} workflows using synthetic datasets constructed to mirror real-world patterns (e.g., TAT-QA): multiple \textit{questions} are associated with a shared \textit{context} and processed by agents with fixed \textit{system prompts}. We define the workload configuration as a quadruple of token lengths: \textit{system prompt}/\textit{context}/\textit{question}/\textit{output}.

\autoref{fig:prefix-sensitivity} illustrates the speedup relative to \emph{LangGraph} on a batch of 100 queries. In prefill-dominated, high-sharing scenarios (P1, P4), \name achieves substantial gains (up to 2.07× on \emph{Debate}) by eliminating redundant prefill computation and enhancing batching efficiency via proactive caching and cache-aware scheduling. Even in decode-oriented configurations or scenarios with highly divergent prefixes (P2, P5), \name maintains a consistent advantage. Notably, in scenarios with relatively short system prompts but shared contexts (P3, P6), \name's scheduler can identify and exploit the shared context across questions. This confirms \name's robustness in less favorable scenarios, demonstrating its ability to optimize diverse prompt structures beyond static system prompts.

\begin{figure}[t]
    \centering
    \includegraphics[width=\linewidth]{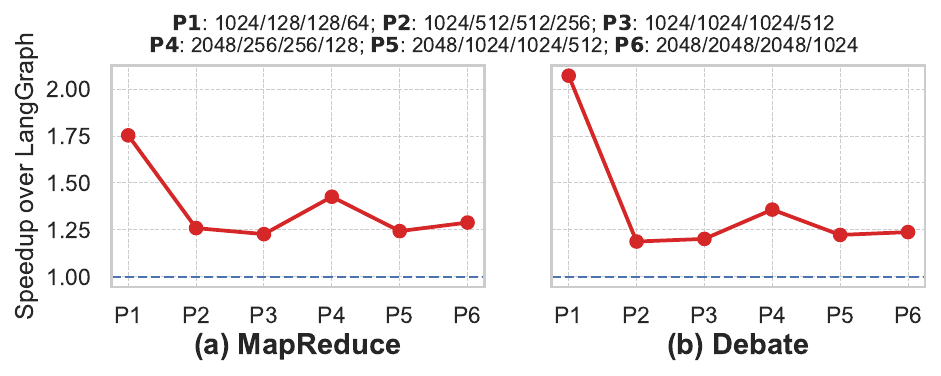}\vspace{-0.15in}
    \caption{Speedup of \name over \emph{LangGraph} on (a) \emph{MapRed} and (b) \emph{Debate} across prompt configurations, denoted by the token lengths of \textit{system prompt}/\textit{context}/\textit{question}/\textit{output}.}
    \label{fig:prefix-sensitivity}
\end{figure}

\subsection{Overhead Analysis}

To better understand the performance characteristics of \name, we conduct an overhead analysis to quantify the latencies introduced by key system components and the evaluation environment compared to the baselines.

\vspace{0.05in}\noindent\textbf{OpenAI API and Load Balancer}.
As detailed in Section~\ref{sec:microbenchmarks}, the baselines interact with the inference backend via the OpenAI API, managed by a lightweight load balancer. To verify that this architecture does not unfairly penalize the baselines, we profile the communication and routing latencies on the \emph{Trading} workflow (batch size 16, 1,408 total LLM calls). Our measurements show that the load balancer adds a minimal 89\,$\mu$s per request, while the average HTTP round-trip latency is as low as 3.1\,ms due to the client and server residing on the same machine. Furthermore, since the workload is dominated by GPU computation, these overheads are effectively masked by concurrency and remain negligible compared to the total end-to-end execution time.

\vspace{0.05in}\noindent\textbf{Component Latency Breakdown}.
To isolate the overhead of \name's internal components, including query optimization (QO), cache-aware scheduling (CAS), and TRT construction, we decompose the end-to-end latency on the \emph{Trading} workflow using Qwen3-8B with a batch size of 16. The results in \autoref{fig:overhead-analysis}(a) demonstrate that execution time is strictly dominated by query processing (QP). Even as workflow complexity increases to 16 branches, the combined planning overhead (QO, CAS, and TRT) remains negligible ($<$230\,ms), confirming that \name's optimization phase is highly efficient relative to the cost of LLM inference.

\vspace{0.05in}\noindent\textbf{Memory Footprint}.
We compare the peak memory usage of \name's scheduling structures (TRT and metadata) against SGLang's RadixCache metadata. \autoref{fig:overhead-analysis}(b) reveals a significant disparity: at 16 branches, \name consumes only 552\,KiB versus SGLang's 14.8\,MiB. This efficiency stems from our abstraction; \name's TRT scales with the workflow structure (number of operators and static prompt templates), whereas SGLang's RadixCache grows linearly with the number of requests and decoded tokens, resulting in significantly higher memory pressure.

\begin{figure}[t]
    \centering
    \includegraphics[width=\linewidth]{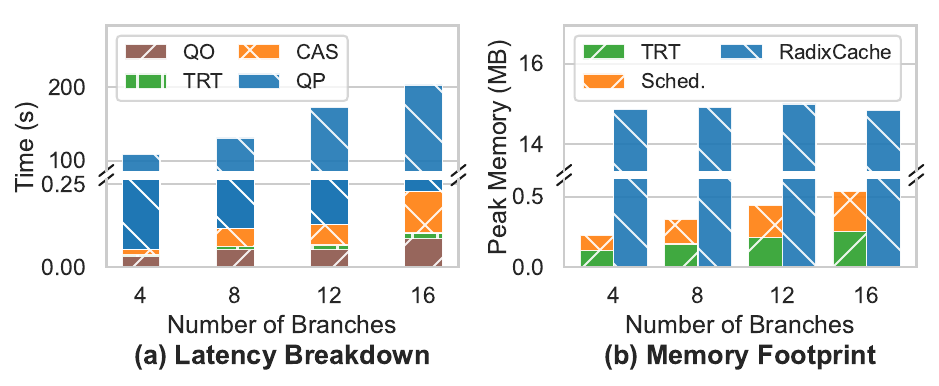}\vspace{-0.15in}
    \caption{Overhead analysis: (a) Latency breakdown: planning (QO, CAS, TRT) vs. query processing (QP); (b) Memory footprint of scheduling metadata.}
    \label{fig:overhead-analysis}
\end{figure}

\subsection{Case Study}
To provide a deeper insight into the dynamic behavior of \name's optimization and scheduling strategies, we conduct a case study on the \emph{Trading} workflow with a batch size of 16.

\vspace{0.05in}\noindent\textbf{Effective LLM Batch Sizes and Latency}.
To understand how \name's scheduling impacts execution dynamics, we trace the execution on two LLM workers, monitoring both inference batch sizes and per-request end-to-end latency. We observe two key batching metrics over time: the number of requests in each inference batch and the total number of tokens (including shared prefixes). To isolate the impact of scheduling, we disable proactive prompt caching for \name and compare its performance against \emph{LangGraph}.

The inference batching results, shown in \autoref{fig:case-study}(a) and (b), demonstrate that \name's cache-aware scheduling leads to more efficient batching. On average, \name processes 1.18× more requests per batch compared to \emph{LangGraph}. This improved batching directly contributes to better GPU utilization and lower end-to-end latency. Furthermore, \name achieves a 1.41× higher peak and a 1.18× higher average number of effective batched tokens. By scheduling requests with shared prefixes consecutively on the same worker, \name's scheduler not only maximizes KV cache reuse but also increases the effective batch size. This confirms that \name's scheduling strategy translates cache efficiency into larger, more effective batches, further improving overall system throughput.

Finally, we analyze how this improved throughput impacts the per-request latency distribution (\autoref{fig:case-study}(c)). \name demonstrates a significant advantage over \emph{LangGraph}, improving the median latency from 28.3\,s to 20.5\,s. Notably, the gap widens at the tail: \name reduces the 95th percentile latency to 37.2\,s compared to 51.7\,s for \emph{LangGraph}. This confirms that cache-aware scheduling effectively mitigates congestion, preventing the straggler requests that typically degrade total workflow completion time.

\begin{figure}[t]
    \centering
    \includegraphics[width=\linewidth]{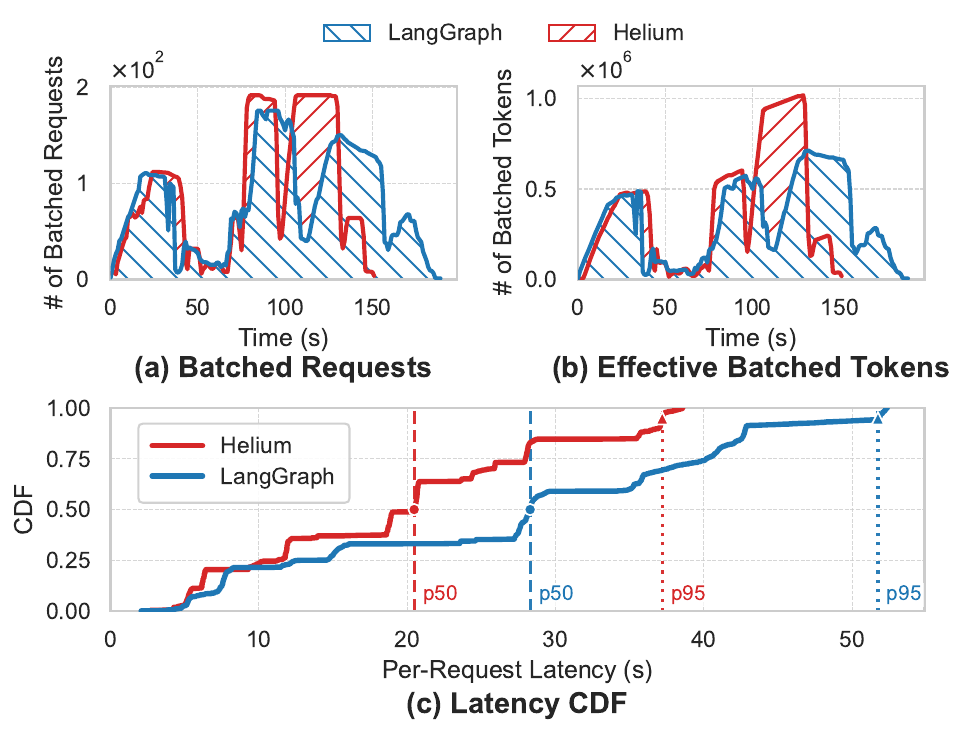}\vspace{-0.15in}
    \caption{Execution dynamics: (a) Batched requests; (b) Effective tokens; (c) Per-request latency CDF. \name's batch efficiency significantly reduces tail latency.}
    \label{fig:case-study}
\end{figure}

\vspace{0.05in}\noindent\textbf{Optimization and Scheduling in Action}. To illustrate \name's optimization and scheduling strategies operate in practice, we trace the execution of the \emph{Trading} workflow. \autoref{fig:optimization} visualizes this process, focusing on a simplified subgraph of the \emph{analyst} stage for clarity. The process begins with the initial workflow DAG (\autoref{fig:optimization}a). First, \name's query optimizer probes the proactive prompt cache and identifies that the outputs for the fundamentals and social media agents are already cached. It rewrites the logical plan by replacing these subgraphs with \texttt{CacheFetch} operators, effectively pruning two branches from the execution graph (\autoref{fig:optimization}b).

The query processor then receives this optimized logical plan, constructs a TRT to capture the prefix structure (\autoref{fig:optimization}c), and applies its cache-aware scheduling algorithm. The resulting schedule (\autoref{fig:optimization}d) demonstrates how the system balances parallelism and prefix cache reuse. The scheduler first groups the execution of \texttt{Op1} and \texttt{Op2} (market agent) to maximize the reuse of their shared prompt prefix across the query batch. Then, instead of idling while waiting for the dependent operator (\texttt{Op3}), the scheduler interleaves the independent operators from the news agent (\texttt{Op4} and \texttt{Op5}). This cost-based decision hides the dependency latency of \texttt{Op3} and maximizes GPU utilization. Then, \texttt{Op3} and \texttt{Op6} are scheduled once their inputs are ready. This demonstrates how \name's holistic approach produces an efficient execution plan that is unattainable by systems that lack a global view of the workflow.

\begin{figure}[t]
    \centering
    \includegraphics[width=\linewidth]{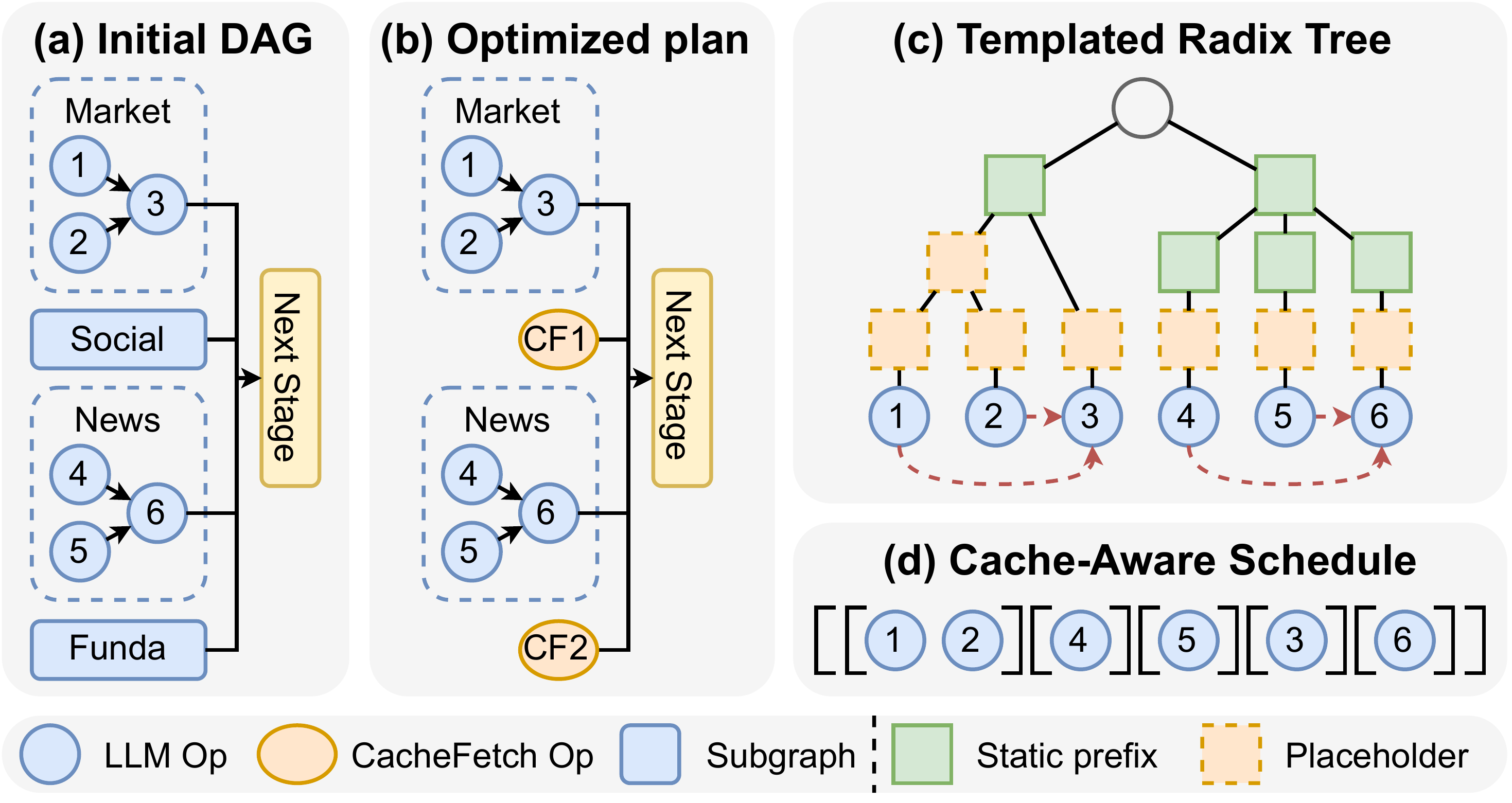}\vspace{-0.1in}
    \caption{\name's optimization and scheduling process: (a) initial DAG, (b) optimized logical plan, (c) TRT construction, and (d) cache-aware schedule.}
    \label{fig:optimization}
\end{figure}

%% file: sections/8_related_work.tex
\section{Related Work}
\vspace{0.05in}\noindent\textbf{LLM Inference Optimizations}.
LLM inference optimizations span multiple levels. Kernel and system-side work enhances throughput and utilization~\cite{flashattention, flashattention-2, flashattention-3, flashinfer, pod-attention, fasttree, orca, sarathi-serve, batchllm}. KV cache management and parallelism enable larger contexts and scaling~\cite{vllm, vattention, megatron-lm}, while disaggregated serving reduce interference and tailors resources~\cite{splitwise, distserve}. Resource multiplexing across tuning and serving further improves efficiency~\cite{he2025resource, he2024deferred}. Speculative decoding accelerates generation~\cite{leviathan2023specdec, miao2024specinfer, cai2024medusa, shen2025hybriddrafting}. Techniques like SGLang~\cite{sglang} and others~\cite{prompt-cache, cache-blend, chunkattention, infinigen} enhance KV cache reuse. However, these workflow-agnostic stacks cannot eliminate agentic workflow redundancy or schedule across queries for maximum reuse.

\vspace{0.05in}\noindent\textbf{LLM-based Agentic Workflows}.
Agentic workflows tackle complex tasks across domains such as software engineering, social simulation, and finance~\cite{llm-mas-survey, metagpt, chatdev, generative-agents, econagent, asfm, cryptotrade, tradingagents, elliottagents, fin-vision}. Frameworks like LangChain, LangGraph, and AutoGen simplify development and orchestration but miss system-level optimizations~\cite{langchain, langgraph, autogen}. Recent systems like Ayo~\cite{ayo}, Parrot~\cite{parrot}, Autellix~\cite{autellix}, KVFlow~\cite{kvflow}, Halo~\cite{shen2025halo}, and FlowMesh~\cite{shen2025flowmesh} introduce workflow awareness to improve scheduling, caching, or composability. Yet, none of these target eliminating prompt redundancy across queries or performing workflow-level, cache-aware scheduling to maximize prefix reuse, as \name does.

\vspace{0.05in}\noindent\textbf{LLMs in Data Systems}.
Data systems integrate LLMs through UDFs and SQL extensions for batch analytics~\cite{spark, databricks}. Prior works like Palimpzest~\cite{palimpzest} and others~\cite{optimizing-llm-queries} optimize LLM-driven analytics via cost models or prompt reordering, often trading off quality. In contrast, \name improves agentic workflow efficiency by reconciling workflow-agnostic serving with database query optimization. By treating LLMs as white-box operators, \name enables proactive caching and cache-aware scheduling to minimize redundant computation while preserving exact semantics.

%% file: sections/9_limitations.tex
\section{Limitations and Future Work}

While Helium establishes a foundation for workflow-aware LLM serving, our current design choices prioritize cache locality and scheduling efficiency. We briefly discuss their implications and future directions.

\vspace{0.05in}\noindent\textbf{Expressiveness of DAG Abstraction}. Our DAG-based abstraction enables powerful optimizations but faces challenges with dynamic control flows common in emerging agentic patterns, such as conditional looping and dynamic mapping, where downstream fan-out depends on runtime outputs. While traditional DAGs struggle with these structures~\cite{zhang2025aflow, kvflow, airflow}, integrating control flow primitives (e.g., as in TensorFlow~\cite{abadi2016tensorflow}) offers a path to statically capture them. We plan to explore these extensions in future work.

\vspace{0.05in}\noindent\textbf{External API Calls}. Integrating external tools (e.g., web search, databases) is critical for agentic workflows but introduces unpredictable latencies that complicate our cost-based scheduling. Currently, \name handles these via best-effort execution. Future work could incorporate stochastic cost models to better account for external variability, enabling more robust global optimization even with black-box API dependencies.

%% file: sections/10_conclusion.tex
\section{Conclusion}

This paper introduces Helium, a workflow-aware serving system that models agentic LLM workloads as query plans with LLMs as first-class operators. By combining proactive caching with cost-based, cache-aware scheduling, Helium removes redundant computation and boosts hardware efficiency. Experiments show substantial speedups over existing frameworks, demonstrating that applying query-optimization principles to LLM serving enables scalable, end-to-end efficiency for agentic AI systems.

%% file: sections/appendix.tex
\clearpage
\appendix

\section{Templated Radix Tree Construction}
\label{sec:trt-construction}

The templated radix tree (TRT) is built from the optimized workflow DAG to capture the shared prefix structure and data dependencies among all LLM operators. The construction process, outlined in Algorithm~\ref{alg:trt-construction}, translates the operator-level graph into the TRT's prefix-based representation.

First, we perform a topological sort of the workflow's operators to ensure that an operator's dependencies are processed before the operator itself (Line~\ref{lst:line:topological-sort}). The algorithm then iterates through the sorted operators (Line~\ref{lst:line:iterate}). For each LLM operator, it constructs a \textit{prefix template}, a sequence of static tokens and dynamic placeholders that represent inputs from antecedent operators (Line~\ref{lst:line:get-prefix-template}). This template is generated by recursively traversing the operator's inputs within the DAG: static strings are appended directly, while dependencies on other operators are encoded as placeholders.

With the prefix template and its set of LLM dependencies identified (Line 8), the operator is inserted into the TRT (Line~\ref{lst:line:get-llm-dependencies}). The \texttt{tree.add} method implements standard radix tree insertion, traversing from the root to find the longest common prefix. If the template diverges from an existing path, the relevant node is split to create a new branch. The LLM operator is then added as a new leaf node. Its dependencies on other LLM operators are recorded as directed edges that connect it to the corresponding leaf nodes retrieved from the \texttt{node\_map}. This map maintains the correspondence between operator IDs and their nodes in the TRT, ensuring that dependencies are correctly resolved as the tree is built (Line~\ref{lst:line:node-map}).

The time complexity of this procedure is determined by the initial graph traversal and the subsequent insertion of each LLM operator into the tree. Let $N$ be the number of operators and $|E|$ be the number of dependency edges in the workflow DAG. The topological sort takes $O(N + |E|)$ time. The total work to generate all prefix templates is also bounded by $O(N + |E|)$, as it involves traversing the DAG. Let $N_{llm} \leq N$ be the number of LLM operators and $L$ be the maximum length of a prefix template. Inserting one template takes $O(L)$ time. The total complexity is therefore $O(N + |E| + N_{llm} \cdot L)$. In a connected DAG, $|E|$ is at least $N-1$, so the complexity simplifies to $O(|E| + N_{llm} \cdot L)$.

\begin{algorithm}[h]
\caption{Templated Radix Tree Construction}
\label{alg:trt-construction}
    \begin{algorithmic}[1]
        \State \textbf{Input:} A workflow DAG, workflow\_dag; a worker assignment map, worker\_assignment.
        \State \textbf{Output:} A templated radix tree, tree.
        \Statex
        \State tree $\gets$ \Call{TemplatedRadixTree}{\null}
        \State node\_map $\gets$ empty map from operator ID to TRT node
        \State sorted\_ops $\gets$ \Call{TopologicalSort}{workflow\_dag} \label{lst:line:topological-sort}
        \For{op \textbf{in} sorted\_ops} \label{lst:line:iterate}
            \If{op \textbf{is} LLM op}
                \State template $\gets$ \Call{GetPrefixTemplate}{op} \label{lst:line:get-prefix-template}
                \State deps $\gets$ \Call{GetLLMDependencies}{op, node\_map} \label{lst:line:get-llm-dependencies}
                \State worker $\gets$ worker\_assignment[op.id]
                \State node $\gets$ tree.\Call{Add}{template, worker, op.id, deps}
                \State node\_map[op.id] $\gets$ node \label{lst:line:node-map}
            \EndIf
        \EndFor
        \State \textbf{return} tree
    \end{algorithmic}
\end{algorithm}

\section{Complexity Analysis of the Scheduling Algorithm}\label{sec:complexity-analysis}

In this section, we provide a formal proof for the time complexity of our cache-aware scheduling algorithm, as detailed in Algorithm~\ref{alg:scheduling-concise}. We demonstrate that the algorithm has a polynomial time complexity in the size of the TRT and the dependency graph.

\vspace{0.05in}\noindent\textbf{Definitions}
Let $T=(V, E)$ be the input TRT, where $|V|$ is the number of nodes. Let $L \subset V$ be the set of leaf nodes, which represent the LLM operators to be scheduled. The dependencies among these operators are given by the DAG $G=(L, E')$, where $|E'|$ is the number of dependency edges. Let $|V_{int}| = |V| - |L| - 1$ be the number of internal nodes in $T$, let $d = \text{depth}(T)$ be the maximum depth of the TRT, and let $c_{max}$ be the maximum branching factor (i.e., number of children) of any node in $T$.

\vspace{0.05in}\noindent\textbf{Proof of Complexity}
\begin{theorem}
    The time complexity of the scheduling algorithm (Algorithm 1) is $O(|V_{int}| \cdot c_{max}^3 + |E'| \cdot d)$.
\end{theorem}
\begin{proof}
    The algorithm consists of two primary phases: an initialization phase and a recursive scheduling phase. We analyze the complexity of each phase separately.
    
    \vspace{0.05in}\noindent\textbf{1. Initialization Phase}:
    The initialization phase constructs the internal scheduling tree data structure from the input TRT and the dependency graph $G$. This involves three main steps:
    \begin{enumerate}
        \item Traversing the TRT ($T$) to create corresponding scheduling nodes, which takes $O(|V|)$ time.
        \item Processing the dependency graph $G=(L, E')$ to compute the depth of each LLM operator (leaf node). This is equivalent to a topological sort on $G$, which costs $O(|L| + |E'|)$.
        \item Building the sibling dependency maps by performing a post-order traversal of the TRT. During this traversal, each dependency in $E'$ is examined to determine if it is internal to a subtree or connects sibling subtrees. This step has a complexity of $O(|V| + |E'|)$.
    \end{enumerate}
    Therefore, the time complexity of the initialization phase is dominated by the tree and graph traversals, resulting in $O(|V| + |E'|)$.
    
    \vspace{0.05in}\noindent\textbf{2. Recursive Scheduling Phase}:
    The core of the algorithm lies in the \texttt{Schedule} function, which invokes the recursive function \texttt{Recurse}. The complexity is dominated by the work done within all calls to \texttt{Recurse}. We analyze the work done at a single internal node $u \in V_{int}$. The work at an internal node $u$ is driven by the \texttt{while} loop that iterates until all of its children are scheduled. Let $c_u$ be the number of children of node $u$. The loop runs at most $c_u$ times. Within each iteration, the most computationally intensive operations are \texttt{SelectBestChild} and the subsequent state updates.
    \begin{itemize}[leftmargin=*]
        \item \textbf{Cost of \texttt{SelectBestChild}}: This function implements our cost-based heuristic. To determine the optimal child to schedule next (based on progress and readiness), the function must analyze the inter-dependencies among the currently unscheduled sibling subtrees. This requires constructing a temporary dependency graph among the $k$ remaining children and analyzing its structure (e.g., by finding strongly connected components to determine topological depth). Such an analysis has a complexity of at least $O(k^2)$. The total work for the \texttt{while} loop at node $u$ is therefore the sum of this cost over its iterations: $$ \sum_{k=1}^{c_u} O(k^2) = O(c_u^3) $$ Summing this cost over all internal nodes in the tree gives the first term of our complexity: $O(|V_{int}| \cdot c_{max}^3)$.
        \item \textbf{Cost of State Updates}: After scheduling a child, the \texttt{UpdateState} function is called on the remaining sibling children. This propagates the effects of newly resolved dependencies. Over the entire execution of the algorithm, a single dependency edge in $E'$ may be re-evaluated each time a scheduling decision is made that could affect its status. The maximum number of times an edge can be re-evaluated is proportional to its depth in the TRT, which is bounded by $d$. Therefore, the total cumulative work for all state updates across the entire algorithm is bounded by $O(|E'| \cdot d)$. 
    \end{itemize}
    Combining the costs from the initialization and scheduling phases, the dominant factor is the recursive scheduling phase. The total time complexity is the sum of the costs of iterative child selection and cumulative state updates. $$ O(|V| + |E'|) + O(|V_{int}| \cdot c_{max}^3 + |E'| \cdot d) $$ As $|V|$ and $|E'|$ are subsumed by the terms in the scheduling phase for any non-trivial workflow, the final complexity is $O(|V_{int}| \cdot c_{max}^3 + |E'| \cdot d)$.
\end{proof}

\section{Agentic Workflows DSL}
\label{sec:appendix-dsl}

\name's DSL allows users to define agentic workflows using a concise set of primitive operators, listed in \autoref{tab:ops}. Designed as a lazy, symbolic dataflow system, the DSL separates workflow definition from execution.
\begin{itemize}[leftmargin=*]
    \item \textbf{IO Operators}: \texttt{Input} and \texttt{Output} define the workflow's interface, allowing batches of data to be injected and results to be retrieved. \texttt{Data} nodes encapsulate static batch data embedded directly within the workflow DAG.
    \item \textbf{Prompt Operators}: \texttt{Format} handles string interpolation for prompt templating, while \texttt{Lambda} allows users to inject arbitrary Python code for stateless transformations (e.g., parsing LLM outputs).
    \item \textbf{Model Operators}: The \texttt{LLM} operator encapsulates calls to the inference engine. It accepts formatted prompts and conversation history, returning generated responses that flow to downstream operators.
\end{itemize}
By composing these primitives, users build a symbolic DAG that the \name runtime analyzes to perform global optimizations like proactive caching and cache-aware scheduling.

\begin{table}[h]\centering
\footnotesize
\begin{tabular}{p{0.08\columnwidth}p{0.08\columnwidth}p{0.68\columnwidth}}
\toprule
\textbf{Op} & \textbf{Type} & \textbf{Description} \\
\midrule
Data & IO &  Encapsulates a data batch as a node in the DAG. \\
Input & IO & Placeholder for data provided at compile time. \\
Output & IO & Materializes data as a final output of the workflow. \\
Format & Prompt & Formats input text using specified arguments. \\
Lambda & Prompt & Applies an arbitrary Python function to the input data. \\
LLM & Model & Generates LLM responses for input conversation contexts. \\
\bottomrule
\end{tabular}
\caption{Examples of Helium's primitive operators.}\label{tab:ops}
\end{table}

\section{Prefix Cache Utilization}\label{sec:prefix-cache-util}
In this section, we describe additional prefix cache utilization results following our end-to-end benchmark. We evaluate \name and the baselines on the \emph{Trading} workflow with a batch size of 16 using the Qwen3-8B and Qwen3-14B models. The results are summarized in \autoref{tab:cache-hit-rate}.

\begin{table}[h]
    \footnotesize
    \begin{tabular}{cll}
    \toprule
    \multirow{2}{*}{\textbf{System}} & \multicolumn{2}{c}{\textbf{Cache hit rate (\%)}}                      \\ \cline{2-3} 
                            & \multicolumn{1}{c}{\textbf{Qwen3-8B}} & \multicolumn{1}{c}{\textbf{Qwen3-14B}} \\ \midrule
    vLLM                    & 42.9 (-27.5\%)               & 43.0 (+0.3\%)                 \\
    OpWise                  & 40.8 (-31.1\%)               & 39.9 (-6.9\%)                 \\
    LangGraph               & 37.1 (-37.3\%)               & 36.7 (-14.4\%)                \\
    AgentScope              & 39.3 (-33.7\%)               & 36.5 (-14.9\%)                \\
    Parrot                  & 41.8 (-29.4\%)               & 36.9 (-13.8\%)                \\
    KVFlow                  & 39.1 (-33.9\%)               & 40.8 (-4.9\%)                 \\ \hline
    \name                   & 59.2                         & 42.9                          \\ \bottomrule
    \end{tabular}
    \caption{Prefix-cache hit rates on the \emph{Trading} workflow (batch size = 16) for Qwen3-8B and Qwen3-14B, with relative differences vs \name.}
    \label{tab:cache-hit-rate}
\end{table}

Overall, \name achieves the highest prefix cache hit rates for Qwen3-8B across all baselines, and is competitive for Qwen3-14B. For Qwen3-8B, \name outperforms all baselines by 27.5\% to 37.3\%. The largest gains are over \emph{LangGraph} (37.3\%) and \emph{KVFlow} (33.9\%), both of which lack the global scheduling awareness to align requests with shared prefixes. \emph{OpWise} achieves even lower hit rate than \emph{vLLM} for Qwen3-8B, as the \emph{Trading} workflow exhibits more prefix sharing across operators (e.g., system prompts shared across agents) than across queries, making operator-by-operator execution especially harmful for cache reuse.

For Qwen3-14B, \name's advantage is more modest (up to 14.9\%), and \emph{vLLM} achieves a marginally higher hit rate (+0.3\%) than \name. The larger KV cache footprint of Qwen3-14B increases memory pressure, which limits \name's ability to retain all precomputed prefixes and reduces the effectiveness of proactive caching relative to the smaller model. Importantly, however, \emph{vLLM}'s competitive hit rate does not translate into competitive end-to-end latency: as shown in Section~\ref{sec:end-to-end}, \emph{vLLM} incurs extremely high latency because it executes each query's operators sequentially, foregoing the batch computation that \name's scheduler exploits. This confirms that a high cache hit rate alone is insufficient: global workflow-aware scheduling is essential to translate cache reuse into end-to-end performance gains.

\emph{Parrot} trails \name by 29.4\% for Qwen3-8B. Although \emph{Parrot} dispatches requests sharing the same prefix to the same engine to maximize cache hits, its routing heuristic creates severe load imbalance across workers. This congestion increases memory pressure on overloaded engines, triggering more frequent KV cache evictions that paradoxically reduce cache hit rates. The result is that \emph{Parrot}'s strategy undermines its own prefix-reuse objective, leading to both lower hit rates and significantly higher end-to-end latency as shown in Section~\ref{sec:end-to-end}. These results collectively highlight the advantage of \name's cost-based, globally-aware approach to cache-aware scheduling.

%% file: references.bib
@inproceedings {parrot,
author = {Chaofan Lin and Zhenhua Han and Chengruidong Zhang and Yuqing Yang and Fan Yang and Chen Chen and Lili Qiu},
title = {Parrot: Efficient Serving of {LLM-based} Applications with Semantic Variable},
booktitle = {18th USENIX Symposium on Operating Systems Design and Implementation (OSDI 24)},
year = {2024},
isbn = {978-1-939133-40-3},
address = {Santa Clara, CA},
pages = {929--945},
url = {https://www.usenix.org/conference/osdi24/presentation/lin-chaofan},
publisher = {USENIX Association},
month = jul
}

@inproceedings{multiagent-debate,
author = {Du, Yilun and Li, Shuang and Torralba, Antonio and Tenenbaum, Joshua B. and Mordatch, Igor},
title = {Improving factuality and reasoning in language models through multiagent debate},
year = {2024},
publisher = {JMLR.org},
booktitle = {Proceedings of the 41st International Conference on Machine Learning},
articleno = {467},
numpages = {31},
location = {Vienna, Austria},
series = {ICML'24}
}

@article{agentscope-simulation,
  publtype={informal},
  author={Xuchen Pan and Dawei Gao and Yuexiang Xie and Zhewei Wei and Yaliang Li and Bolin Ding and Ji-Rong Wen and Jingren Zhou},
  title={Very Large-Scale Multi-Agent Simulation in AgentScope},
  year={2024},
  cdate={1704067200000},
  journal={CoRR},
  volume={abs/2407.17789},
  url={https://doi.org/10.48550/arXiv.2407.17789}
}

@article{spark,
author = {Zaharia, Matei and Xin, Reynold S. and Wendell, Patrick and Das, Tathagata and Armbrust, Michael and Dave, Ankur and Meng, Xiangrui and Rosen, Josh and Venkataraman, Shivaram and Franklin, Michael J. and Ghodsi, Ali and Gonzalez, Joseph and Shenker, Scott and Stoica, Ion},
title = {Apache Spark: a unified engine for big data processing},
year = {2016},
issue_date = {November 2016},
publisher = {Association for Computing Machinery},
address = {New York, NY, USA},
volume = {59},
number = {11},
issn = {0001-0782},
url = {https://doi.org/10.1145/2934664},
doi = {10.1145/2934664},
journal = {Commun. ACM},
month = oct,
pages = {56–65},
numpages = {10}
}

@inproceedings{
mmlu,
title={Measuring Massive Multitask Language Understanding},
author={Dan Hendrycks and Collin Burns and Steven Basart and Andy Zou and Mantas Mazeika and Dawn Song and Jacob Steinhardt},
booktitle={International Conference on Learning Representations},
year={2021},
url={https://openreview.net/forum?id=d7KBjmI3GmQ}
}

@inproceedings{tatqa,
    title = "{TAT}-{QA}: A Question Answering Benchmark on a Hybrid of Tabular and Textual Content in Finance",
    author = "Zhu, Fengbin  and
      Lei, Wenqiang  and
      Huang, Youcheng  and
      Wang, Chao  and
      Zhang, Shuo  and
      Lv, Jiancheng  and
      Feng, Fuli  and
      Chua, Tat-Seng",
    editor = "Zong, Chengqing  and
      Xia, Fei  and
      Li, Wenjie  and
      Navigli, Roberto",
    booktitle = "Proceedings of the 59th Annual Meeting of the Association for Computational Linguistics and the 11th International Joint Conference on Natural Language Processing (Volume 1: Long Papers)",
    month = aug,
    year = "2021",
    address = "Online",
    publisher = "Association for Computational Linguistics",
    url = "https://aclanthology.org/2021.acl-long.254/",
    doi = "10.18653/v1/2021.acl-long.254",
    pages = "3277--3287"
}

@misc{iterative-sum,
	author = {LangChain},
	title = {{H}ow to summarize text through iterative refinement | {L}ang{C}hain --- python.langchain.com},
	howpublished = {\url{https://python.langchain.com/docs/how_to/summarize_refine/}},
	year = {[n.d.]},
	note = {[Accessed 26-08-2025]},
}

@article{amazon-reviews,
  title={Bridging Language and Items for Retrieval and Recommendation},
  author={Hou, Yupeng and Li, Jiacheng and He, Zhankui and Yan, An and Chen, Xiusi and McAuley, Julian},
  journal={arXiv preprint arXiv:2403.03952},
  year={2024}
}

@misc{qwen3,
      title={Qwen3 Technical Report}, 
      author={An Yang and Anfeng Li and Baosong Yang and Beichen Zhang and Binyuan Hui and Bo Zheng and Bowen Yu and Chang Gao and Chengen Huang and Chenxu Lv and Chujie Zheng and Dayiheng Liu and Fan Zhou and Fei Huang and Feng Hu and Hao Ge and Haoran Wei and Huan Lin and Jialong Tang and Jian Yang and Jianhong Tu and Jianwei Zhang and Jianxin Yang and Jiaxi Yang and Jing Zhou and Jingren Zhou and Junyang Lin and Kai Dang and Keqin Bao and Kexin Yang and Le Yu and Lianghao Deng and Mei Li and Mingfeng Xue and Mingze Li and Pei Zhang and Peng Wang and Qin Zhu and Rui Men and Ruize Gao and Shixuan Liu and Shuang Luo and Tianhao Li and Tianyi Tang and Wenbiao Yin and Xingzhang Ren and Xinyu Wang and Xinyu Zhang and Xuancheng Ren and Yang Fan and Yang Su and Yichang Zhang and Yinger Zhang and Yu Wan and Yuqiong Liu and Zekun Wang and Zeyu Cui and Zhenru Zhang and Zhipeng Zhou and Zihan Qiu},
      year={2025},
      eprint={2505.09388},
      archivePrefix={arXiv},
      primaryClass={cs.CL},
      url={https://arxiv.org/abs/2505.09388}, 
}

@inproceedings{vllm,
author = {Kwon, Woosuk and Li, Zhuohan and Zhuang, Siyuan and Sheng, Ying and Zheng, Lianmin and Yu, Cody Hao and Gonzalez, Joseph and Zhang, Hao and Stoica, Ion},
title = {Efficient Memory Management for Large Language Model Serving with PagedAttention},
year = {2023},
isbn = {9798400702297},
publisher = {Association for Computing Machinery},
address = {New York, NY, USA},
url = {https://doi.org/10.1145/3600006.3613165},
doi = {10.1145/3600006.3613165},
booktitle = {Proceedings of the 29th Symposium on Operating Systems Principles},
pages = {611–626},
numpages = {16},
location = {Koblenz, Germany},
series = {SOSP '23}
}

@article{liu2025supporting,
  title={Supporting Our AI Overlords: Redesigning Data Systems to be Agent-First},
  author={Liu, Shu and Ponnapalli, Soujanya and Shankar, Shreya and Zeighami, Sepanta and Zhu, Alan and Agarwal, Shubham and Chen, Ruiqi and Suwito, Samion and Yuan, Shuo and Stoica, Ion and Zaharia, Matei and Cheung, Alvin and Crooks, Natacha and Gonzalez, Joseph E. and Parameswaran, Aditya G.},
  journal={arXiv preprint arXiv:2509.00997},
  year={2025}
}

@misc{vllm-github,
  title={{vLLM: Easy, Fast, and Cheap LLM Serving}},
  author={{vLLM Team}},
  howpublished={\url{https://github.com/vllm-project/vllm}},
  year={2023},
  note={Accessed: 2025-09-15}
}

@misc{langchain,
  author    = {LangChain},
  title     = {LangChain: Building applications with LLMs through composable abstractions},
  howpublished = {\url{https://python.langchain.com/docs/introduction/}},
  year      = {[n.d.]},
  note      = {Accessed: 2025-09-16}
}

@misc{langgraph,
  author    = {LangChain},
  title     = {LangGraph: Graph-based orchestration for LLM workflows},
  howpublished = {\url{https://langchain-ai.github.io/langgraph/}},
  year      = {[n.d.]},
  note      = {Accessed: 2025-09-16}
}

@article{agentscope,
  author    = {Dawei Gao and Zitao Li and Xuchen Pan and Weirui Kuang and Zhijian Ma and Bingchen Qian and Fei Wei and Wenhao Zhang and Yuexiang Xie and Daoyuan Chen and Liuyi Yao and Hongyi Peng and Zeyu Zhang and Lin Zhu and Chen Cheng and Hongzhu Shi and Yaliang Li and Bolin Ding and Jingren Zhou},
  title     = {AgentScope: A Flexible yet Robust Multi-Agent Platform},
  journal   = {arXiv preprint arXiv:2402.14034},
  year      = {2024},
  url       = {https://arxiv.org/abs/2402.14034}
}

@inproceedings{liu2024agentic,
  author    = {Yilun Du and Shuang Li and Antonio Torralba and Joshua B. Tenenbaum and Igor Mordatch},
  title     = {Improving factuality and reasoning in language models through multi-agent debate},
  booktitle = {Proceedings of the 41st International Conference on Machine Learning (ICML)},
  year      = {2024},
  url       = {https://arxiv.org/abs/2402.14034}
}

@inproceedings{sglang,
author = {Zheng, Lianmin and Yin, Liangsheng and Xie, Zhiqiang and Sun, Chuyue and Huang, Jeff and Yu, Cody Hao and Cao, Shiyi and Kozyrakis, Christos and Stoica, Ion and Gonzalez, Joseph E. and Barrett, Clark and Sheng, Ying},
title = {SGLang: efficient execution of structured language model programs},
year = {2025},
isbn = {9798331314385},
publisher = {Curran Associates Inc.},
address = {Red Hook, NY, USA},
booktitle = {Proceedings of the 38th International Conference on Neural Information Processing Systems},
articleno = {2000},
numpages = {27},
location = {Vancouver, BC, Canada},
series = {NIPS '24}
}

@inproceedings{critical-path,
author = {Kelley, James E. and Walker, Morgan R.},
title = {Critical-path planning and scheduling},
year = {1959},
isbn = {9781450378680},
publisher = {Association for Computing Machinery},
address = {New York, NY, USA},
url = {https://doi.org/10.1145/1460299.1460318},
doi = {10.1145/1460299.1460318},
booktitle = {Papers Presented at the December 1-3, 1959, Eastern Joint IRE-AIEE-ACM Computer Conference},
pages = {160–173},
numpages = {14},
location = {Boston, Massachusetts},
series = {IRE-AIEE-ACM '59 (Eastern)}
}

@inproceedings{cache-blend,
  author = {Yao, Jiayi and Li, Hanchen and Liu, Yuhan and Ray, Siddhant and Cheng, Yihua and Zhang, Qizheng and Du, Kuntai and Lu, Shan and Jiang, Junchen},
  title = {CacheBlend: Fast Large Language Model Serving for RAG with Cached Knowledge Fusion},
  year = {2025},
  url = {https://doi.org/10.1145/3689031.3696098},
  doi = {10.1145/3689031.3696098},
  booktitle = {Proceedings of the Twentieth European Conference on Computer Systems},
  pages = {94–109},
}

@ARTICLE{llm-mas-survey,
  title     = "A survey on {LLM-based} multi-agent systems: workflow,
               infrastructure, and challenges",
  author    = "Li, Xinyi and Wang, Sai and Zeng, Siqi and Wu, Yu and Yang, Yi",
  journal   = "Vicinagearth",
  publisher = "Springer Science and Business Media LLC",
  volume    =  1,
  number    =  1,
  month     =  oct,
  year      =  2024,
  copyright = "https://creativecommons.org/licenses/by/4.0",
  language  = "en"
}

@inproceedings{
metagpt,
title={Meta{GPT}: Meta Programming for A Multi-Agent Collaborative Framework},
author={Sirui Hong and Mingchen Zhuge and Jonathan Chen and Xiawu Zheng and Yuheng Cheng and Jinlin Wang and Ceyao Zhang and Zili Wang and Steven Ka Shing Yau and Zijuan Lin and Liyang Zhou and Chenyu Ran and Lingfeng Xiao and Chenglin Wu and J{\"u}rgen Schmidhuber},
booktitle={The Twelfth International Conference on Learning Representations},
year={2024},
url={https://openreview.net/forum?id=VtmBAGCN7o}
}

@inproceedings{chatdev,
    title = "{C}hat{D}ev: Communicative Agents for Software Development",
    author = "Qian, Chen  and
      Liu, Wei  and
      Liu, Hongzhang  and
      Chen, Nuo  and
      Dang, Yufan  and
      Li, Jiahao  and
      Yang, Cheng  and
      Chen, Weize  and
      Su, Yusheng  and
      Cong, Xin  and
      Xu, Juyuan  and
      Li, Dahai  and
      Liu, Zhiyuan  and
      Sun, Maosong",
    editor = "Ku, Lun-Wei  and
      Martins, Andre  and
      Srikumar, Vivek",
    booktitle = "Proceedings of the 62nd Annual Meeting of the Association for Computational Linguistics (Volume 1: Long Papers)",
    month = aug,
    year = "2024",
    address = "Bangkok, Thailand",
    publisher = "Association for Computational Linguistics",
    url = "https://aclanthology.org/2024.acl-long.810/",
    doi = "10.18653/v1/2024.acl-long.810",
    pages = "15174--15186"
}

@inproceedings{generative-agents,
author = {Park, Joon Sung and O'Brien, Joseph and Cai, Carrie Jun and Morris, Meredith Ringel and Liang, Percy and Bernstein, Michael S.},
title = {Generative Agents: Interactive Simulacra of Human Behavior},
year = {2023},
isbn = {9798400701320},
publisher = {Association for Computing Machinery},
address = {New York, NY, USA},
url = {https://doi.org/10.1145/3586183.3606763},
doi = {10.1145/3586183.3606763},
booktitle = {Proceedings of the 36th Annual ACM Symposium on User Interface Software and Technology},
articleno = {2},
numpages = {22},
location = {San Francisco, CA, USA},
series = {UIST '23}
}

@inproceedings{econagent,
    title = "{E}con{A}gent: Large Language Model-Empowered Agents for Simulating Macroeconomic Activities",
    author = "Li, Nian  and
      Gao, Chen  and
      Li, Mingyu  and
      Li, Yong  and
      Liao, Qingmin",
    editor = "Ku, Lun-Wei  and
      Martins, Andre  and
      Srikumar, Vivek",
    booktitle = "Proceedings of the 62nd Annual Meeting of the Association for Computational Linguistics (Volume 1: Long Papers)",
    month = aug,
    year = "2024",
    address = "Bangkok, Thailand",
    publisher = "Association for Computational Linguistics",
    url = "https://aclanthology.org/2024.acl-long.829/",
    doi = "10.18653/v1/2024.acl-long.829",
    pages = "15523--15536"
}

@misc{asfm,
      title={Simulating Financial Market via Large Language Model based Agents}, 
      author={Shen Gao and Yuntao Wen and Minghang Zhu and Jianing Wei and Yuhan Cheng and Qunzi Zhang and Shuo Shang},
      year={2024},
      eprint={2406.19966},
      archivePrefix={arXiv},
      primaryClass={cs.CL},
      url={https://arxiv.org/abs/2406.19966}, 
}

@inproceedings{cryptotrade,
    title = "{C}rypto{T}rade: A Reflective {LLM}-based Agent to Guide Zero-shot Cryptocurrency Trading",
    author = "Li, Yuan  and
      Luo, Bingqiao  and
      Wang, Qian  and
      Chen, Nuo  and
      Liu, Xu  and
      He, Bingsheng",
    editor = "Al-Onaizan, Yaser  and
      Bansal, Mohit  and
      Chen, Yun-Nung",
    booktitle = "Proceedings of the 2024 Conference on Empirical Methods in Natural Language Processing",
    month = nov,
    year = "2024",
    address = "Miami, Florida, USA",
    publisher = "Association for Computational Linguistics",
    url = "https://aclanthology.org/2024.emnlp-main.63/",
    doi = "10.18653/v1/2024.emnlp-main.63",
    pages = "1094--1106"
}

@inproceedings{
tradingagents,
title={TradingAgents: Multi-Agents {LLM} Financial Trading Framework},
author={Yijia Xiao and Edward Sun and Di Luo and Wei Wang},
booktitle={The First MARW: Multi-Agent AI in the Real World Workshop at AAAI 2025},
year={2025},
url={https://openreview.net/forum?id=4QPrXwMQt1}
}

@inproceedings{elliottagents,
    title = "{E}lliott{A}gents: A Natural Language-Driven Multi-Agent System for Stock Market Analysis and Prediction",
    author = "Chudziak, Jaroslaw A.  and
      Wawer, Michal",
    editor = "Oco, Nathaniel  and
      Dita, Shirley N.  and
      Borlongan, Ariane Macalinga  and
      Kim, Jong-Bok",
    booktitle = "Proceedings of the 38th Pacific Asia Conference on Language, Information and Computation",
    month = dec,
    year = "2024",
    address = "Tokyo, Japan",
    publisher = "Tokyo University of Foreign Studies",
    url = "https://aclanthology.org/2024.paclic-1.91/",
    pages = "961--970"
}

@inproceedings{fin-vision,
author = {Fatemi, Sorouralsadat and Hu, Yuheng},
title = {FinVision: A Multi-Agent Framework for Stock Market Prediction},
year = {2024},
isbn = {9798400710810},
publisher = {Association for Computing Machinery},
address = {New York, NY, USA},
url = {https://doi.org/10.1145/3677052.3698688},
doi = {10.1145/3677052.3698688},
booktitle = {Proceedings of the 5th ACM International Conference on AI in Finance},
pages = {582–590},
numpages = {9},
location = {Brooklyn, NY, USA},
series = {ICAIF '24}
}

@inproceedings{
autogen,
title={AutoGen: Enabling Next-Gen {LLM} Applications via Multi-Agent Conversations},
author={Qingyun Wu and Gagan Bansal and Jieyu Zhang and Yiran Wu and Beibin Li and Erkang Zhu and Li Jiang and Xiaoyun Zhang and Shaokun Zhang and Jiale Liu and Ahmed Hassan Awadallah and Ryen W White and Doug Burger and Chi Wang},
booktitle={First Conference on Language Modeling},
year={2024},
url={https://openreview.net/forum?id=BAakY1hNKS}
}

@inbook{ayo,
author = {Tan, Xin and Jiang, Yimin and Yang, Yitao and Xu, Hong},
title = {Towards End-to-End Optimization of LLM-based Applications with Ayo},
year = {2025},
isbn = {9798400710797},
publisher = {Association for Computing Machinery},
address = {New York, NY, USA},
url = {https://doi.org/10.1145/3676641.3716278},
booktitle = {Proceedings of the 30th ACM International Conference on Architectural Support for Programming Languages and Operating Systems, Volume 2},
pages = {1302–1316},
numpages = {15}
}

@misc{autellix,
      title={Autellix: An Efficient Serving Engine for LLM Agents as General Programs}, 
      author={Michael Luo and Xiaoxiang Shi and Colin Cai and Tianjun Zhang and Justin Wong and Yichuan Wang and Chi Wang and Yanping Huang and Zhifeng Chen and Joseph E. Gonzalez and Ion Stoica},
      year={2025},
      eprint={2502.13965},
      archivePrefix={arXiv},
      primaryClass={cs.LG},
      url={https://arxiv.org/abs/2502.13965}, 
}

@misc{kvflow,
      title={KVFlow: Efficient Prefix Caching for Accelerating LLM-Based Multi-Agent Workflows}, 
      author={Zaifeng Pan and Ajjkumar Patel and Zhengding Hu and Yipeng Shen and Yue Guan and Wan-Lu Li and Lianhui Qin and Yida Wang and Yufei Ding},
      year={2025},
      eprint={2507.07400},
      archivePrefix={arXiv},
      primaryClass={cs.DC},
      url={https://arxiv.org/abs/2507.07400}, 
}

@inproceedings {orca,
author = {Gyeong-In Yu and Joo Seong Jeong and Geon-Woo Kim and Soojeong Kim and Byung-Gon Chun},
title = {Orca: A Distributed Serving System for {Transformer-Based} Generative Models},
booktitle = {16th USENIX Symposium on Operating Systems Design and Implementation (OSDI 22)},
year = {2022},
isbn = {978-1-939133-28-1},
address = {Carlsbad, CA},
pages = {521--538},
url = {https://www.usenix.org/conference/osdi22/presentation/yu},
publisher = {USENIX Association},
month = jul
}

@inproceedings{flashattention,
 author = {Dao, Tri and Fu, Dan and Ermon, Stefano and Rudra, Atri and R\'{e}, Christopher},
 booktitle = {Advances in Neural Information Processing Systems},
 editor = {S. Koyejo and S. Mohamed and A. Agarwal and D. Belgrave and K. Cho and A. Oh},
 pages = {16344--16359},
 publisher = {Curran Associates, Inc.},
 title = {FlashAttention: Fast and Memory-Efficient Exact Attention with IO-Awareness},
 url = {https://proceedings.neurips.cc/paper_files/paper/2022/file/67d57c32e20fd0a7a302cb81d36e40d5-Paper-Conference.pdf},
 volume = {35},
 year = {2022}
}

@inproceedings{flashattention-2,
title={FlashAttention-2: Faster Attention with Better Parallelism and Work Partitioning},
author={Tri Dao},
booktitle={The Twelfth International Conference on Learning Representations},
year={2024},
url={https://openreview.net/forum?id=mZn2Xyh9Ec}
}

@inproceedings{flashattention-3,
 author = {Shah, Jay and Bikshandi, Ganesh and Zhang, Ying and Thakkar, Vijay and Ramani, Pradeep and Dao, Tri},
 booktitle = {Advances in Neural Information Processing Systems},
 editor = {A. Globerson and L. Mackey and D. Belgrave and A. Fan and U. Paquet and J. Tomczak and C. Zhang},
 pages = {68658--68685},
 publisher = {Curran Associates, Inc.},
 title = {FlashAttention-3: Fast and Accurate Attention with Asynchrony and Low-precision},
 url = {https://proceedings.neurips.cc/paper_files/paper/2024/file/7ede97c3e082c6df10a8d6103a2eebd2-Paper-Conference.pdf},
 volume = {37},
 year = {2024}
}

@article{flashinfer,
  publtype={informal},
  author={Zihao Ye and Lequn Chen and Ruihang Lai and Wuwei Lin and Yineng Zhang and Stephanie Wang and Tianqi Chen and Baris Kasikci and Vinod Grover and Arvind Krishnamurthy and Luis Ceze},
  title={FlashInfer: Efficient and Customizable Attention Engine for LLM Inference Serving},
  year={2025},
  month={January},
  cdate={1735689600000},
  journal={CoRR},
  volume={abs/2501.01005},
  url={https://doi.org/10.48550/arXiv.2501.01005}
}

@inproceedings{
fasttree,
title={FastTree: Optimizing Attention Kernel and Runtime for Tree-Structured {LLM} Inference},
author={Zaifeng Pan and Yitong Ding and Yue Guan and Zheng Wang and Zhongkai Yu and Xulong Tang and Yida Wang and Yufei Ding},
booktitle={Eighth Conference on Machine Learning and Systems},
year={2025},
url={https://openreview.net/forum?id=BwvHcHZ3kJ}
}

@inbook{pod-attention,
author = {Kamath, Aditya K. and Prabhu, Ramya and Mohan, Jayashree and Peter, Simon and Ramjee, Ramachandran and Panwar, Ashish},
title = {POD-Attention: Unlocking Full Prefill-Decode Overlap for Faster LLM Inference},
year = {2025},
isbn = {9798400710797},
publisher = {Association for Computing Machinery},
address = {New York, NY, USA},
url = {https://doi.org/10.1145/3676641.3715996},
booktitle = {Proceedings of the 30th ACM International Conference on Architectural Support for Programming Languages and Operating Systems, Volume 2},
pages = {897–912},
numpages = {16}
}

@inproceedings{vattention,
author = {Prabhu, Ramya and Nayak, Ajay and Mohan, Jayashree and Ramjee, Ramachandran and Panwar, Ashish},
title = {vAttention: Dynamic Memory Management for Serving LLMs without PagedAttention},
year = {2025},
isbn = {9798400706981},
publisher = {Association for Computing Machinery},
address = {New York, NY, USA},
url = {https://doi.org/10.1145/3669940.3707256},
doi = {10.1145/3669940.3707256},
booktitle = {Proceedings of the 30th ACM International Conference on Architectural Support for Programming Languages and Operating Systems, Volume 1},
pages = {1133–1150},
numpages = {18},
location = {Rotterdam, Netherlands},
series = {ASPLOS '25}
}

@inproceedings{sarathi-serve,
author = {Agrawal, Amey and Kedia, Nitin and Panwar, Ashish and Mohan, Jayashree and Kwatra, Nipun and Gulavani, Bhargav S. and Tumanov, Alexey and Ramjee, Ramachandran},
title = {Taming throughput-latency tradeoff in LLM inference with sarathi-serve},
year = {2024},
isbn = {978-1-939133-40-3},
publisher = {USENIX Association},
address = {USA},
booktitle = {Proceedings of the 18th USENIX Conference on Operating Systems Design and Implementation},
articleno = {7},
numpages = {18},
location = {Santa Clara, CA, USA},
series = {OSDI'24}
}

@INPROCEEDINGS{splitwise,
  author={Patel, Pratyush and Choukse, Esha and Zhang, Chaojie and Shah, Aashaka and Goiri, Íñigo and Maleki, Saeed and Bianchini, Ricardo},
  booktitle={2024 ACM/IEEE 51st Annual International Symposium on Computer Architecture (ISCA)}, 
  title={Splitwise: Efficient Generative LLM Inference Using Phase Splitting}, 
  year={2024},
  volume={},
  number={},
  pages={118-132},
  doi={10.1109/ISCA59077.2024.00019}}

@inproceedings {distserve,
author = {Yinmin Zhong and Shengyu Liu and Junda Chen and Jianbo Hu and Yibo Zhu and Xuanzhe Liu and Xin Jin and Hao Zhang},
title = {{DistServe}: Disaggregating Prefill and Decoding for Goodput-optimized Large Language Model Serving},
booktitle = {18th USENIX Symposium on Operating Systems Design and Implementation (OSDI 24)},
year = {2024},
isbn = {978-1-939133-40-3},
address = {Santa Clara, CA},
pages = {193--210},
url = {https://www.usenix.org/conference/osdi24/presentation/zhong-yinmin},
publisher = {USENIX Association},
month = jul
}

@misc{batchllm,
      title={BatchLLM: Optimizing Large Batched LLM Inference with Global Prefix Sharing and Throughput-oriented Token Batching}, 
      author={Zhen Zheng and Xin Ji and Taosong Fang and Fanghao Zhou and Chuanjie Liu and Gang Peng},
      year={2025},
      eprint={2412.03594},
      archivePrefix={arXiv},
      primaryClass={cs.CL},
      url={https://arxiv.org/abs/2412.03594}, 
}

@misc{megatron-lm,
      title={Megatron-LM: Training Multi-Billion Parameter Language Models Using Model Parallelism}, 
      author={Mohammad Shoeybi and Mostofa Patwary and Raul Puri and Patrick LeGresley and Jared Casper and Bryan Catanzaro},
      year={2020},
      eprint={1909.08053},
      archivePrefix={arXiv},
      primaryClass={cs.CL},
      url={https://arxiv.org/abs/1909.08053}, 
}

@inproceedings{chunkattention,
    title = "{C}hunk{A}ttention: Efficient Self-Attention with Prefix-Aware {KV} Cache and Two-Phase Partition",
    author = "Ye, Lu  and
      Tao, Ze  and
      Huang, Yong  and
      Li, Yang",
    editor = "Ku, Lun-Wei  and
      Martins, Andre  and
      Srikumar, Vivek",
    booktitle = "Proceedings of the 62nd Annual Meeting of the Association for Computational Linguistics (Volume 1: Long Papers)",
    month = aug,
    year = "2024",
    address = "Bangkok, Thailand",
    publisher = "Association for Computational Linguistics",
    url = "https://aclanthology.org/2024.acl-long.623/",
    doi = "10.18653/v1/2024.acl-long.623",
    pages = "11608--11620"
}

@inproceedings{prompt-cache,
 author = {Gim, In and Chen, Guojun and Lee, Seung-seob and Sarda, Nikhil and Khandelwal, Anurag and Zhong, Lin},
 booktitle = {Proceedings of Machine Learning and Systems},
 editor = {P. Gibbons and G. Pekhimenko and C. De Sa},
 pages = {325--338},
 title = {Prompt Cache: Modular Attention Reuse for Low-Latency Inference},
 url = {https://proceedings.mlsys.org/paper_files/paper/2024/file/a66caa1703fe34705a4368c3014c1966-Paper-Conference.pdf},
 volume = {6},
 year = {2024}
}

@misc{databricks,
	author = {Databricks},
	title = {{D}atabricks: {L}eading {D}ata and {A}{I} {S}olutions for {E}nterprises --- databricks.com},
	howpublished = {\url{https://www.databricks.com/}},
	year = {},
	note = {[Accessed 22-09-2025]},
}

@article{palimpzest,
  publtype={informal},
  author={Chunwei Liu and Matthew Russo and Michael J. Cafarella and Lei Cao and Peter Baile Chen and Zui Chen and Michael J. Franklin and Tim Kraska and Samuel Madden and Gerardo Vitagliano},
  title={A Declarative System for Optimizing AI Workloads},
  year={2024},
  cdate={1704067200000},
  journal={CoRR},
  volume={abs/2405.14696},
  url={https://doi.org/10.48550/arXiv.2405.14696}
}

@inproceedings{
optimizing-llm-queries,
title={Optimizing {LLM} Queries in Relational Data Analytics Workloads},
author={Shu Liu and Asim Biswal and Amog Kamsetty and Audrey Cheng and Luis Gaspar Schroeder and Liana Patel and Shiyi Cao and Xiangxi Mo and Ion Stoica and Joseph E. Gonzalez and Matei Zaharia},
booktitle={Eighth Conference on Machine Learning and Systems},
year={2025},
url={https://openreview.net/forum?id=R7bK9yycHp}
}

@inproceedings{infinigen,
author = {Wonbeom Lee and Jungi Lee and Junghwan Seo and Jaewoong Sim},
title = {{InfiniGen}: Efficient Generative Inference of Large Language Models with Dynamic {KV} Cache Management},
booktitle = {18th USENIX Symposium on Operating Systems Design and Implementation (OSDI 24)},
year = {2024},
isbn = {978-1-939133-40-3},
address = {Santa Clara, CA},
pages = {155--172},
url = {https://www.usenix.org/conference/osdi24/presentation/lee},
publisher = {USENIX Association},
month = jul
}

@MISC{harris2002trade-behavior,
  title     = "Trading and exchanges",
  author    = "Harris, Larry",
  publisher = "Oxford University PressNew York, NY",
  month     =  oct,
  year      =  2002
}

@misc{finnhub, title={Finnhub - free realtime apis for stock, Forex and cryptocurrency.}, url={https://finnhub.io/}, journal={Finnhub Stock APIs - Real-time stock prices, Company fundamentals, Estimates, and Alternative data.}, author={Finnhub.io}}

@misc{yfin, url={https://finance.yahoo.com/}, journal={Yahoo! Finance}, publisher={Yahoo!}, author={Yahoo!}}

@misc{reddit-finance, title={Winddude/reddit\_finance\_43\_250k · datasets at hugging face}, url={https://huggingface.co/datasets/winddude/reddit\_finance\_43\_250k}, journal={winddude/reddit\_finance\_43\_250k · Datasets at Hugging Face}, author={Stewart, Lawrence}}

@inproceedings{rocklin2015dask,
  author = {Rocklin, Matthew},
  booktitle = {SciPy},
  editor = {Huff, Kathryn and Bergstra, James},
  ee = {https://doi.org/10.25080/Majora-7b98e3ed-013},
  pages = {126-132},
  publisher = {scipy.org},
  title = {Dask: Parallel Computation with Blocked algorithms and Task Scheduling.},
  url = {http://dblp.uni-trier.de/db/conf/scipy/scipy2015.html#Rocklin15},
  year = 2015
}

@incollection{machine-scheduling,
title = {Complexity of Machine Scheduling Problems},
editor = {P.L. Hammer and E.L. Johnson and B.H. Korte and G.L. Nemhauser},
series = {Annals of Discrete Mathematics},
publisher = {Elsevier},
volume = {1},
pages = {343-362},
year = {1977},
booktitle = {Studies in Integer Programming},
issn = {0167-5060},
doi = {https://doi.org/10.1016/S0167-5060(08)70743-X},
url = {https://www.sciencedirect.com/science/article/pii/S016750600870743X},
author = {J.K. Lenstra and A.H.G. {Rinnooy Kan} and P. Brucker}
}

@misc{tensorflow,
title={ {TensorFlow}: Large-Scale Machine Learning on Heterogeneous Systems},
url={https://www.tensorflow.org/},
note={Software available from tensorflow.org},
author={
    Mart\'{i}n~Abadi and
    Ashish~Agarwal and
    Paul~Barham and
    Eugene~Brevdo and
    Zhifeng~Chen and
    Craig~Citro and
    Greg~S.~Corrado and
    Andy~Davis and
    Jeffrey~Dean and
    Matthieu~Devin and
    Sanjay~Ghemawat and
    Ian~Goodfellow and
    Andrew~Harp and
    Geoffrey~Irving and
    Michael~Isard and
    Yangqing Jia and
    Rafal~Jozefowicz and
    Lukasz~Kaiser and
    Manjunath~Kudlur and
    Josh~Levenberg and
    Dandelion~Man\'{e} and
    Rajat~Monga and
    Sherry~Moore and
    Derek~Murray and
    Chris~Olah and
    Mike~Schuster and
    Jonathon~Shlens and
    Benoit~Steiner and
    Ilya~Sutskever and
    Kunal~Talwar and
    Paul~Tucker and
    Vincent~Vanhoucke and
    Vijay~Vasudevan and
    Fernanda~Vi\'{e}gas and
    Oriol~Vinyals and
    Pete~Warden and
    Martin~Wattenberg and
    Martin~Wicke and
    Yuan~Yu and
    Xiaoqiang~Zheng},
  year={2015},
}

@article{luo2025agent-survey,
  publtype={informal},
  author={Junyu Luo and Weizhi Zhang and Ye Yuan and Yusheng Zhao and Junwei Yang and Yiyang Gu and Bohan Wu and Binqi Chen and Ziyue Qiao and Qingqing Long and Rongcheng Tu and Xiao Luo and Wei Ju and Zhiping Xiao and Yifan Wang and Meng Xiao and Chenwu Liu and Jingyang Yuan and Shichang Zhang and Yiqiao Jin and Fan Zhang and Xian Wu and Hanqing Zhao and Dacheng Tao and Philip S. Yu and Ming Zhang},
  title={Large Language Model Agent: A Survey on Methodology, Applications and Challenges},
  year={2025},
  month={March},
  cdate={1740787200000},
  journal={CoRR},
  volume={abs/2503.21460},
  url={https://doi.org/10.48550/arXiv.2503.21460}
}

@article{xi2025rise-agent,
  author={Zhiheng Xi and Wenxiang Chen and Xin Guo and Wei He and Yiwen Ding and Boyang Hong and Ming Zhang and Junzhe Wang and Senjie Jin and Enyu Zhou and Rui Zheng and Xiaoran Fan and Xiao Wang and Limao Xiong and Yuhao Zhou and Weiran Wang and Changhao Jiang and Yicheng Zou and Xiangyang Liu and Zhangyue Yin and Shihan Dou and Rongxiang Weng and Wenjuan Qin and Yongyan Zheng and Xipeng Qiu and Xuanjing Huang and Qi Zhang and Tao Gui},
  title={The rise and potential of large language model based agents: a survey},
  year={2025},
  cdate={1735689600000},
  journal={Sci. China Inf. Sci.},
  volume={68},
  number={2},
  url={https://doi.org/10.1007/s11432-024-4222-0}
}

@article{lei2024autonomous-agents,
  author={Lei Wang and Chen Ma and Xueyang Feng and Zeyu Zhang and Hao Yang and Jingsen Zhang and Zhiyuan Chen and Jiakai Tang and Xu Chen and Yankai Lin and Wayne Xin Zhao and Zhewei Wei and Jirong Wen},
  title={A survey on large language model based autonomous agents},
  year={2024},
  month={December},
  cdate={1733011200000},
  journal={Frontiers Comput. Sci.},
  volume={18},
  number={6},
  pages={186345},
  url={https://doi.org/10.1007/s11704-024-40231-1}
}

@INPROCEEDINGS{singh2024enhancing-workflows,
  author={Singh, Aditi and Ehtesham, Abul and Kumar, Saket and Khoei, Tala Talaei},
  booktitle={2024 IEEE World AI IoT Congress (AIIoT)}, 
  title={Enhancing AI Systems with Agentic Workflows Patterns in Large Language Model}, 
  year={2024},
  volume={},
  number={},
  pages={527-532},
  doi={10.1109/AIIoT61789.2024.10578990}}

@inproceedings{
zhang2025aflow,
title={{AF}low: Automating Agentic Workflow Generation},
author={Jiayi Zhang and Jinyu Xiang and Zhaoyang Yu and Fengwei Teng and Xiong-Hui Chen and Jiaqi Chen and Mingchen Zhuge and Xin Cheng and Sirui Hong and Jinlin Wang and Bingnan Zheng and Bang Liu and Yuyu Luo and Chenglin Wu},
booktitle={The Thirteenth International Conference on Learning Representations},
year={2025},
url={https://openreview.net/forum?id=z5uVAKwmjf}
}

@INPROCEEDINGS{yu2025workflow-survey,
  author={Yu, Chaojia and Cheng, Zihan and Cui, Hanwen and Gao, Yishuo and Luo, Zexu and Wang, Yijin and Zheng, Hangbin and Zhao, Yong},
  booktitle={2025 8th International Conference on Artificial Intelligence and Big Data (ICAIBD)}, 
  title={A Survey on Agent Workflow – Status and Future}, 
  year={2025},
  volume={},
  number={},
  pages={770-781},
  doi={10.1109/ICAIBD64986.2025.11082076}}

@article{sapkota2025agents-vs-agentic,
   title={AI Agents vs. Agentic AI: A Conceptual taxonomy, applications and challenges},
   volume={126},
   ISSN={1566-2535},
   url={http://dx.doi.org/10.1016/j.inffus.2025.103599},
   DOI={10.1016/j.inffus.2025.103599},
   journal={Information Fusion},
   publisher={Elsevier BV},
   author={Sapkota, Ranjan and Roumeliotis, Konstantinos I. and Karkee, Manoj},
   year={2026},
   month=feb, pages={103599} }

@misc{asgar2025agentic-hetero,
      title={Efficient and Scalable Agentic AI with Heterogeneous Systems}, 
      author={Zain Asgar and Michelle Nguyen and Sachin Katti},
      year={2025},
      eprint={2507.19635},
      archivePrefix={arXiv},
      primaryClass={cs.LG},
      url={https://arxiv.org/abs/2507.19635}, 
}

@misc{xi2025deepsearch-survey,
      title={A Survey of LLM-based Deep Search Agents: Paradigm, Optimization, Evaluation, and Challenges}, 
      author={Yunjia Xi and Jianghao Lin and Yongzhao Xiao and Zheli Zhou and Rong Shan and Te Gao and Jiachen Zhu and Weiwen Liu and Yong Yu and Weinan Zhang},
      year={2025},
      eprint={2508.05668},
      archivePrefix={arXiv},
      primaryClass={cs.IR},
      url={https://arxiv.org/abs/2508.05668}, 
}

@misc{tensorrt-llm, title={Welcome to TENSORRT LLM’s documentation!}, url={https://nvidia.github.io/TensorRT-LLM/}, journal={Welcome to TensorRT LLM’s Documentation! - TensorRT LLM}, author={NVIDIA}}

@inproceedings{
wang2023selfconsistency,
title={Self-Consistency Improves Chain of Thought Reasoning in Language Models},
author={Xuezhi Wang and Jason Wei and Dale Schuurmans and Quoc V Le and Ed H. Chi and Sharan Narang and Aakanksha Chowdhery and Denny Zhou},
booktitle={The Eleventh International Conference on Learning Representations },
year={2023},
url={https://openreview.net/forum?id=1PL1NIMMrw}
}

@inproceedings{yao2023tree-of-thoughts,
 author = {Yao, Shunyu and Yu, Dian and Zhao, Jeffrey and Shafran, Izhak and Griffiths, Tom and Cao, Yuan and Narasimhan, Karthik},
 booktitle = {Advances in Neural Information Processing Systems},
 editor = {A. Oh and T. Naumann and A. Globerson and K. Saenko and M. Hardt and S. Levine},
 pages = {11809--11822},
 publisher = {Curran Associates, Inc.},
 title = {Tree of Thoughts: Deliberate Problem Solving with Large Language Models},
 url = {https://proceedings.neurips.cc/paper_files/paper/2023/file/271db9922b8d1f4dd7aaef84ed5ac703-Paper-Conference.pdf},
 volume = {36},
 year = {2023}
}

@inproceedings{liang2024mad,
    title = "Encouraging Divergent Thinking in Large Language Models through Multi-Agent Debate",
    author = "Liang, Tian  and
      He, Zhiwei  and
      Jiao, Wenxiang  and
      Wang, Xing  and
      Wang, Yan  and
      Wang, Rui  and
      Yang, Yujiu  and
      Shi, Shuming  and
      Tu, Zhaopeng",
    editor = "Al-Onaizan, Yaser  and
      Bansal, Mohit  and
      Chen, Yun-Nung",
    booktitle = "Proceedings of the 2024 Conference on Empirical Methods in Natural Language Processing",
    month = nov,
    year = "2024",
    address = "Miami, Florida, USA",
    publisher = "Association for Computational Linguistics",
    url = "https://aclanthology.org/2024.emnlp-main.992/",
    doi = "10.18653/v1/2024.emnlp-main.992",
    pages = "17889--17904"
}

@article{hellerstein2007db-architecture,
author = {Hellerstein, Joseph M. and Stonebraker, Michael and Hamilton, James},
title = {Architecture of a Database System},
year = {2007},
issue_date = {February 2007},
publisher = {Now Publishers Inc.},
address = {Hanover, MA, USA},
volume = {1},
number = {2},
issn = {1931-7883},
url = {https://doi.org/10.1561/1900000002},
doi = {10.1561/1900000002},
journal = {Found. Trends Databases},
month = feb,
pages = {141–259},
numpages = {119}
}

@article{goetz1993query-evaluation,
author = {Graefe, Goetz},
title = {Query evaluation techniques for large databases},
year = {1993},
issue_date = {June 1993},
publisher = {Association for Computing Machinery},
address = {New York, NY, USA},
volume = {25},
number = {2},
issn = {0360-0300},
url = {https://doi.org/10.1145/152610.152611},
doi = {10.1145/152610.152611},
journal = {ACM Comput. Surv.},
month = jun,
pages = {73–169},
numpages = {97}
}

@article{morrison1968patricia,
author = {Morrison, Donald R.},
title = {PATRICIA—Practical Algorithm To Retrieve Information Coded in Alphanumeric},
year = {1968},
issue_date = {Oct. 1968},
publisher = {Association for Computing Machinery},
address = {New York, NY, USA},
volume = {15},
number = {4},
issn = {0004-5411},
url = {https://doi.org/10.1145/321479.321481},
doi = {10.1145/321479.321481},
journal = {J. ACM},
month = oct,
pages = {514–534},
numpages = {21}
}

@inproceedings{besta2024graph-of-thoughts,
author = {Besta, Maciej and Blach, Nils and Kubicek, Ales and Gerstenberger, Robert and Podstawski, Micha\l{} and Gianinazzi, Lukas and Gajda, Joanna and Lehmann, Tomasz and Niewiadomski, Hubert and Nyczyk, Piotr and Hoefler, Torsten},
title = {Graph of thoughts: solving elaborate problems with large language models},
year = {2024},
isbn = {978-1-57735-887-9},
publisher = {AAAI Press},
url = {https://doi.org/10.1609/aaai.v38i16.29720},
doi = {10.1609/aaai.v38i16.29720},
booktitle = {Proceedings of the Thirty-Eighth AAAI Conference on Artificial Intelligence and Thirty-Sixth Conference on Innovative Applications of Artificial Intelligence and Fourteenth Symposium on Educational Advances in Artificial Intelligence},
articleno = {1972},
numpages = {9},
series = {AAAI'24/IAAI'24/EAAI'24}
}

@article{sellis1988optimization,
author = {Sellis, Timos K.},
title = {Multiple-query optimization},
year = {1988},
issue_date = {March 1988},
publisher = {Association for Computing Machinery},
address = {New York, NY, USA},
volume = {13},
number = {1},
issn = {0362-5915},
url = {https://doi.org/10.1145/42201.42203},
doi = {10.1145/42201.42203},
journal = {ACM Trans. Database Syst.},
month = mar,
pages = {23–52},
numpages = {30}
}

@article{graefe1994volcano,
author = {Graefe, G.},
title = {Volcano-An Extensible and Parallel Query Evaluation System},
year = {1994},
issue_date = {February 1994},
publisher = {IEEE Educational Activities Department},
address = {USA},
volume = {6},
number = {1},
issn = {1041-4347},
url = {https://doi.org/10.1109/69.273032},
doi = {10.1109/69.273032},
journal = {IEEE Trans. on Knowl. and Data Eng.},
month = feb,
pages = {120–135},
numpages = {16}
}

@inproceedings{leviathan2023specdec,
author = {Leviathan, Yaniv and Kalman, Matan and Matias, Yossi},
title = {Fast inference from transformers via speculative decoding},
year = {2023},
publisher = {JMLR.org},
booktitle = {Proceedings of the 40th International Conference on Machine Learning},
articleno = {795},
numpages = {13},
location = {Honolulu, Hawaii, USA},
series = {ICML'23}
}

@inproceedings{miao2024specinfer,
author = {Miao, Xupeng and Oliaro, Gabriele and Zhang, Zhihao and Cheng, Xinhao and Wang, Zeyu and Zhang, Zhengxin and Wong, Rae Ying Yee and Zhu, Alan and Yang, Lijie and Shi, Xiaoxiang and Shi, Chunan and Chen, Zhuoming and Arfeen, Daiyaan and Abhyankar, Reyna and Jia, Zhihao},
title = {SpecInfer: Accelerating Large Language Model Serving with Tree-based Speculative Inference and Verification},
year = {2024},
isbn = {9798400703867},
publisher = {Association for Computing Machinery},
address = {New York, NY, USA},
url = {https://doi.org/10.1145/3620666.3651335},
doi = {10.1145/3620666.3651335},
booktitle = {Proceedings of the 29th ACM International Conference on Architectural Support for Programming Languages and Operating Systems, Volume 3},
pages = {932–949},
numpages = {18},
location = {La Jolla, CA, USA},
series = {ASPLOS '24}
}

@inproceedings{cai2024medusa,
author = {Cai, Tianle and Li, Yuhong and Geng, Zhengyang and Peng, Hongwu and Lee, Jason D. and Chen, Deming and Dao, Tri},
title = {MEDUSA: Simple LLM inference acceleration framework with multiple decoding heads},
year = {2024},
publisher = {JMLR.org},
booktitle = {Proceedings of the 41st International Conference on Machine Learning},
articleno = {203},
numpages = {27},
location = {Vienna, Austria},
series = {ICML'24}
}

@misc{shen2025hybriddrafting,
      title={Speculative Decoding via Hybrid Drafting and Rollback-Aware Branch Parallelism}, 
      author={Yuhao Shen and Junyi Shen and Quan Kong and Tianyu Liu and Yao Lu and Cong Wang},
      year={2025},
      eprint={2506.01979},
      archivePrefix={arXiv},
      primaryClass={cs.DC},
      url={https://arxiv.org/abs/2506.01979}, 
}

@misc{shen2025halo,
      title={Batch Query Processing and Optimization for Agentic Workflows}, 
      author={Junyi Shen and Noppanat Wadlom and Yao Lu},
      year={2025},
      eprint={2509.02121},
      archivePrefix={arXiv},
      primaryClass={cs.DB},
      url={https://arxiv.org/abs/2509.02121}, 
}

@MISC{forrest2024cbc,
  title     = "{Coin-or/CBC}: Release releases/2.10.12",
  author    = "Forrest, John and Ralphs, Ted and Vigerske, Stefan and Santos,
               Haroldo Gambini and Forrest, John and Hafer, Lou and
               Kristjansson, Bjarni and {jpfasano} and {EdwinStraver} and
               {Jan-Willem} and Lubin, Miles and {rlougee} and {a-andre} and
               {jpgoncal} and Brito, Samuel and {h-i-gassmann} and {Cristina}
               and Saltzman, Matthew and {tosttost} and Pitrus, Bruno and
               Matsushima, Fumiaki and Vossler, Patrick and Swgy, Ron @ and
               {to-st}",
  publisher = "Zenodo",
  year      =  2024
}

@inproceedings {abadi2016tensorflow,
author = {Mart{\'\i}n Abadi and Paul Barham and Jianmin Chen and Zhifeng Chen and Andy Davis and Jeffrey Dean and Matthieu Devin and Sanjay Ghemawat and Geoffrey Irving and Michael Isard and Manjunath Kudlur and Josh Levenberg and Rajat Monga and Sherry Moore and Derek G. Murray and Benoit Steiner and Paul Tucker and Vijay Vasudevan and Pete Warden and Martin Wicke and Yuan Yu and Xiaoqiang Zheng},
title = {{TensorFlow}: A System for {Large-Scale} Machine Learning},
booktitle = {12th USENIX Symposium on Operating Systems Design and Implementation (OSDI 16)},
year = {2016},
isbn = {978-1-931971-33-1},
address = {Savannah, GA},
pages = {265--283},
url = {https://www.usenix.org/conference/osdi16/technical-sessions/presentation/abadi},
publisher = {USENIX Association},
month = nov
}

@misc{airflow,
url={https://airflow.apache.org/},
journal={Apache Airflow},
year={}
}

@inproceedings{he2025resource,
author = {He, Yongjun and Yang, Haofeng and Lu, Yao and Klimovi\'{c}, Ana and Alonso, Gustavo},
title = {Resource multiplexing in tuning and serving large language models},
year = {2025},
isbn = {978-1-939133-48-9},
publisher = {USENIX Association},
address = {USA},
booktitle = {Proceedings of the 2025 USENIX Conference on Usenix Annual Technical Conference},
articleno = {97},
numpages = {17},
location = {Boston, MA, USA},
series = {USENIX ATC '25}
}

@inproceedings{he2024deferred,
author = {He, Yongjun and Lu, Yao and Alonso, Gustavo},
title = {Deferred Continuous Batching in Resource-Efficient Large Language Model Serving},
year = {2024},
isbn = {9798400705410},
publisher = {Association for Computing Machinery},
address = {New York, NY, USA},
url = {https://doi.org/10.1145/3642970.3655835},
doi = {10.1145/3642970.3655835},
booktitle = {Proceedings of the 4th Workshop on Machine Learning and Systems},
pages = {98–106},
numpages = {9},
location = {Athens, Greece},
series = {EuroMLSys '24}
}

@misc{shen2025flowmesh,
      title={FlowMesh: A Service Fabric for Composable LLM Workflows}, 
      author={Junyi Shen and Noppanat Wadlom and Lingfeng Zhou and Dequan Wang and Xu Miao and Lei Fang and Yao Lu},
      year={2025},
      eprint={2510.26913},
      archivePrefix={arXiv},
      primaryClass={cs.DC},
      url={https://arxiv.org/abs/2510.26913}, 
}
